\documentclass[a4paper,UKenglish]{lipics-v2019}
\nolinenumbers

\usepackage{amsmath,amssymb,amsthm}
\RequirePackage{ifthen}
\RequirePackage{xparse}
\RequirePackage{xcolor}
\RequirePackage{xspace}
\RequirePackage{etoolbox}
\RequirePackage{tabularx}
\RequirePackage{ltablex}
\RequirePackage{amsmath}
\RequirePackage{pgffor}

\provideboolean{PackageBabelAvailable}
\ifcsdef{bbl@loaded}{\setboolean{PackageBabelAvailable}{true}}{}

\ifbool{PackageBabelAvailable}{%
	\RequirePackage{translations}%
}{}

\makeatletter             %
\newcommand{\providecounter}[1]{%
	\ifcsdef{c@#1}%
		{\setcounter{#1}{0}}%
		{\newcounter{#1}}%
}

\provideboolean{@NNUseRTL}
\provideboolean{@NNDetailed}
\provideboolean{@NNInListMode}
\provideboolean{@NNPrintGroup}
\providecounter{@NNCurrentGroup}
\providecounter{@NNNumberOfAllGroups}
\providecounter{@NNNumberOfEntriesInGroup0}%
\providecounter{@NNNumberOfAllEntries}%

\relax

\ifbool{PackageBabelAvailable}{%
	\DeclareTranslationFallback{Table of notation}{Table of notation}%
	\DeclareTranslation{German}{Table of notation}{Notationstabelle}%
	\DeclareTranslationFallback{symbol}{symbol}%
	\DeclareTranslation{German}{symbol}{Symbol}%
	\DeclareTranslationFallback{macro name}{macro name}%
	\DeclareTranslation{German}{macro name}{Makroname}%
	\DeclareTranslationFallback{uses}{uses}%
	\DeclareTranslation{German}{uses}{Verw.}%
	\DeclareTranslationFallback{description}{description}%
	\DeclareTranslation{German}{description}{Beschreibung}%
	\DeclareTranslationFallback{no entries}{no entries}%
	\DeclareTranslation{German}{no entries}{keine Einträge}%
}{}

\ifbool{PackageBabelAvailable}%
	{\newcommand{\gettranslation}[1]{\GetTranslation{#1}}}%
	{\newcommand{\gettranslation}[1]{#1}}

\newcommand{\@NNLAYOUTArgument}[1]{\textcolor{blue}{#1}}
\newcommand{\@NNLAYOUTZeroUses}[1]{\textcolor{red}{#1}}
\newcommand{\@NNLAYOUTColumnHeader}[1]{\textbf{#1}}
\newcommand{\@NNLAYOUTTable}[1]{%
	\notbool{@NNUseRTL}%
		{{\Large\textbf{#1}}\hrulefill}%
		{\hrulefill{\Large\textbf{#1}}}%
}
\newcommand{\@NNLAYOUTGroup}[1]{%
	\notbool{@NNUseRTL}%
		{{\large\textbf{#1}}\dotfill}%
		{{\dotfill{\large\textbf{#1}}}}%
}

\newcommand{\@NNDefaultArgumentA}{1}
\newcommand{\@NNDefaultArgumentB}{2}
\newcommand{\@NNDefaultArgumentC}{3}
\newcommand{\@NNDefaultArgumentD}{4}
\newcommand{\@NNDefaultArgumentE}{5}
\newcommand{\@NNDefaultArgumentF}{6}
\newcommand{\@NNDefaultArgumentG}{7}
\newcommand{\@NNDefaultArgumentH}{8}
\newcommand{\@NNDefaultArgumentI}{9}
\let\@NNArgumentA\@NNDefaultArgumentA
\let\@NNArgumentB\@NNDefaultArgumentB
\let\@NNArgumentC\@NNDefaultArgumentC
\let\@NNArgumentD\@NNDefaultArgumentD
\let\@NNArgumentE\@NNDefaultArgumentE
\let\@NNArgumentF\@NNDefaultArgumentF
\let\@NNArgumentG\@NNDefaultArgumentG
\let\@NNArgumentH\@NNDefaultArgumentH
\let\@NNArgumentI\@NNDefaultArgumentI

\newcommand{\@NNShowIf}[3]{\ifthenelse{#2 > #3}{}{{\@NNLAYOUTArgument{#1}}}}

\newcommand{\@NNPrintTableHeader}{%
	\notbool{@NNDetailed}%
	{%
		\notbool{@NNUseRTL}%
		{%
			\@NNLAYOUTColumnHeader{\gettranslation{symbol}} &%
			\@NNLAYOUTColumnHeader{\gettranslation{description}} \\%
		}%
		{%
			\@NNLAYOUTColumnHeader{\gettranslation{description}} & %
			\@NNLAYOUTColumnHeader{\gettranslation{symbol}} \\ %
		}%
	}%
	{%
		\notbool{@NNUseRTL}%
		{%
				\@NNLAYOUTColumnHeader{\gettranslation{symbol}} &%
				\@NNLAYOUTColumnHeader{\gettranslation{macro name}} &%
				\@NNLAYOUTColumnHeader{\gettranslation{uses}} &%
				\@NNLAYOUTColumnHeader{\gettranslation{description}}\\%
		}%
		{%
				\@NNLAYOUTColumnHeader{\gettranslation{description}}&%
				\@NNLAYOUTColumnHeader{\gettranslation{uses}} &%
				\@NNLAYOUTColumnHeader{\gettranslation{macro name}} &%
				\@NNLAYOUTColumnHeader{\gettranslation{symbol}} \\%
		}%
	}%
}

\newcommand{\@NNPrintCommand}[1]{%
	\letcs{\NNnum}{@NNArgumentsOfEntry#1}%
	\ensuremath{%
		\csname @NNCommandOfEntry#1\endcsname%
		{\@NNShowIf{\@NNArgumentA}{1}{\NNnum}}%
		{\@NNShowIf{\@NNArgumentB}{2}{\NNnum}}%
		{\@NNShowIf{\@NNArgumentC}{3}{\NNnum}}%
		{\@NNShowIf{\@NNArgumentD}{4}{\NNnum}}%
		{\@NNShowIf{\@NNArgumentE}{5}{\NNnum}}%
		{\@NNShowIf{\@NNArgumentF}{6}{\NNnum}}%
		{\@NNShowIf{\@NNArgumentG}{7}{\NNnum}}%
		{\@NNShowIf{\@NNArgumentH}{8}{\NNnum}}%
		{\@NNShowIf{\@NNArgumentI}{9}{\NNnum}}%
	}%
}

\newcommand{\@NNPrintCommandName}[1]{%
	\texttt{%
		\csname @NNMacroNameOfEntry#1\endcsname%
		\expandafter\ifthenelse{\csname @NNArgumentsOfEntry#1\endcsname < 1}{}{%
			[\csname @NNArgumentsOfEntry#1\endcsname]}%
		}%
}

\newcommand{\@NNPrintEntry}[1]{%
	\@nameuse{@NNArgDescsOfEntry#1}%
	\notbool{@NNDetailed}{%
		\notbool{@NNUseRTL}%
			{\@NNPrintCommand{#1} & \csname @NNDescriptionOfEntry#1\endcsname \\}%
			{\csname @NNDescriptionOfEntry#1\endcsname & \@NNPrintCommand{#1} \\}%
	}%
	{%
		\notbool{@NNUseRTL}%
		{%
			\@NNPrintCommand{#1} & %
			\@NNPrintCommandName{#1} & %
			\csname @NNNumberOfUsesOfEntry#1\endcsname & %
			\csname @NNDescriptionOfEntry#1\endcsname \\ %
		}%
		{%
			\csname @NNDescriptionOfEntry#1\endcsname & %
			\csname @NNNumberOfUsesOfEntry#1\endcsname & %
			\@NNPrintCommandName{#1} & %
			\@NNPrintCommand{#1} \\ %
		}%
	}%
}

\newcommand{\@NNPrintGroup}[1]{%
	\setboolean{@NNPrintGroup}{true}
	\ifbool{@NNDefaultGroup#1}%
		{%
			\ifthenelse{\expandafter\value{@NNNumberOfEntriesInGroup#1} < 1}%
				{\setboolean{@NNPrintGroup}{false}}%
				{}%
		}%
		{%
			\noindent\expandafter\@NNLAYOUTGroup{\csname @NNNameOfGroup#1\endcsname}%
			\ifthenelse{\expandafter\value{@NNNumberOfEntriesInGroup#1} < 1}%
				{\setboolean{@NNPrintGroup}{false}\gettranslation{no entries}}%
				{}%
		}%
	\ifbool{@NNPrintGroup}%
	{%
		\keepXColumns%
		\ifbool{@NNDetailed}{%
			\notbool{@NNUseRTL}%
				{\begin{tabularx}{\textwidth}{lllX}}%
				{\begin{tabularx}{\textwidth}{Xlll}}%
		}%
		{%
			\notbool{@NNUseRTL}%
				{\begin{tabularx}{\textwidth}{lX}}%
				{\begin{tabularx}{\textwidth}{Xl}}%
		}%
		\@NNPrintTableHeader %
		\forlistcsloop{\@NNPrintEntry}{@NNEntriesInGroup#1} %
		\end{tabularx}%
	}{}%
}

\newcommand{\@NNGroups}{}

\newcommand{\@NNNewGroup}[3]{%
	\stepcounter{@NNNumberOfAllGroups}%
	\setcounter{@NNCurrentGroup}{\value{@NNNumberOfAllGroups}}%
	\listeadd{\@NNGroups}{\the@NNCurrentGroup}%
	\csedef{@NNNameOfGroup\the@NNCurrentGroup}{#1}%
	\csedef{@NNOutputLevelOfGroup\the@NNCurrentGroup}{#3}%
	\expandafter\provideboolean{@NNDefaultGroup\the@NNCurrentGroup}%
	\expandafter\setboolean{@NNDefaultGroup\the@NNCurrentGroup}{#2}%
	\expandafter\providecounter{@NNNumberOfEntriesInGroup\the@NNCurrentGroup}%
}

\newcommand{\@NNPrepareEntry}[1]{%
	\stepcounter{@NNNumberOfAllEntries}%
	\stepcounter{@NNNumberOfEntriesInGroup\the@NNCurrentGroup}%
	\providecounter{@NNCounter#1}%
	\csdef{@NNNumberOfUsesOfEntry\the@NNNumberOfAllEntries}%
	{%
		\ifthenelse{\value{@NNCounter#1} > 0 }%
			{\arabic{@NNCounter#1}}%
			{\@NNLAYOUTZeroUses{\arabic{@NNCounter#1}}}%
	}%
}

\newcommand{\@NNAddEntry}[5]{%
	\listcseadd{@NNEntriesInGroup\the@NNCurrentGroup}{\the@NNNumberOfAllEntries}%
	\csedef{@NNCommandOfEntry\the@NNNumberOfAllEntries}{#1}%
	\csedef{@NNMacroNameOfEntry\the@NNNumberOfAllEntries}{#2}%
	\csedef{@NNArgumentsOfEntry\the@NNNumberOfAllEntries}{#3}%
	\csedef{@NNDescriptionOfEntry\the@NNNumberOfAllEntries}{#4}%
	\csdef{@NNArgDescsOfEntry\the@NNNumberOfAllEntries}{#5}%
}

\newcommand{\notationSetLayoutTable}[1]{\renewcommand{\@NNLAYOUTTable}[1]{#1}}
\newcommand{\notationSetLayoutGroup}[1]{\renewcommand{\@NNLAYOUTGroup}[1]{#1}}
\newcommand{\notationSetLayoutColumnHeader}[1]{\renewcommand{\@NNLAYOUTColumnHeader}[1]{#1}}
\newcommand{\notationSetLayoutArgument}[1]{\renewcommand{\@NNLAYOUTArgument}[1]{#1}}
\newcommand{\notationSetLayoutZeroUses}[1]{\renewcommand{\@NNLAYOUTZeroUses}[1]{#1}}

\newcommand{\notationarg}[2]{%
	\ifthenelse{#1 = 1}{\renewcommand{\@NNArgumentA}{#2}}{%
	\ifthenelse{#1 = 2}{\renewcommand{\@NNArgumentB}{#2}}{%
	\ifthenelse{#1 = 3}{\renewcommand{\@NNArgumentC}{#2}}{%
	\ifthenelse{#1 = 4}{\renewcommand{\@NNArgumentD}{#2}}{%
	\ifthenelse{#1 = 5}{\renewcommand{\@NNArgumentE}{#2}}{%
	\ifthenelse{#1 = 6}{\renewcommand{\@NNArgumentF}{#2}}{%
	\ifthenelse{#1 = 7}{\renewcommand{\@NNArgumentG}{#2}}{%
	\ifthenelse{#1 = 8}{\renewcommand{\@NNArgumentH}{#2}}{%
	\ifthenelse{#1 = 9}{\renewcommand{\@NNArgumentI}{#2}}%
		{\PackageWarning{newnotation}{notationarg: Invalid argument number}{Must be an integer between 1 and 9.}}%
	}}}}}}}}%
}

\DeclareDocumentCommand{\tableofnotation}{O{0}}{%
	\setboolean{@NNInListMode}{true}%
	{%
		
		\bigskip\noindent%
		\@NNLAYOUTTable{\gettranslation{Table of notation}}

	}%
	\renewcommand*{\do}[1]{%
		\ifthenelse{\@nameuse{@NNOutputLevelOfGroup##1} < #1}{}{%
			\@NNPrintGroup{##1}%
		}%
	}%
	\dolistloop{\@NNGroups}%
	\setboolean{@NNInListMode}{false}%
}

\DeclareDocumentCommand{\detailedtableofnotation}{O{0}}{%
	\setboolean{@NNDetailed}{true}%
	\tableofnotation[#1]%
	\setboolean{@NNDetailed}{false}%
}

\DeclareDocumentCommand{\notationnewgroup}{O{0} m}{%
	\@NNNewGroup{#2}{false}{#1}%
}

\newcommand{\notationnewtable}{%
	\renewcommand{\@NNGroups}{}%
	\@NNNewGroup{}{true}{0}%
}

\notationnewtable %

\newcommand{\notationsavetable}[1]{%
	\csedef{@NNTableCurrentGroup#1}{\the@NNCurrentGroup}%
	\forlistloop{\listcsadd{@NNTable#1}}{\@NNGroups}%
}

\newcommand{\notationloadtable}[1]{%
	\ifcsundef{@NNTable#1}{\PackageError{newnotation}{notationloadtable: Unknown table '#1'.}{Misspelled name?}}{}%
	\notationnewtable%
	\expandafter\setcounter{@NNCurrentGroup}{\csname @NNTableCurrentGroup#1\endcsname}%
	\forlistcsloop{\listcsadd{@NNGroups}}{@NNTable#1}%
}

\newcommand{\newnotationclass}[2]{%
	\csdef{@NNClass#1}{#2}%
}

\DeclareDocumentCommand{\newnotation}{s o m O{0} o m O{} O{}}{%
	\ifdef{#3}%
		{\PackageError{newnotation}{newnotation: Command '\string#3' already defined.}{Misspelled name?}}{}%
	\@NNPrepareEntry{\string#3}%
	\IfNoValueTF{#5}{%
		\newrobustcmd{#3}[#4]{%
			\IfNoValueTF{#2}%
				{\IfBooleanTF{#1}%
					{\ensuremath{#6}}%
					{\ensuremath{{#6}}}}%
				{\IfBooleanTF{#1}%
					{\ensuremath{\@nameuse{@NNClass#2}{#6}}}%
					{\ensuremath{{\@nameuse{@NNClass#2}{#6}}}}}%
			\ifthenelse{\boolean{@NNInListMode}}{}{\protect\stepcounter{@NNCounter\string#3}}%
			\xspace%
		}%
	}{%
	\newrobustcmd{#3}[#4][#5]{%
			\IfNoValueTF{#2}%
				{\IfBooleanTF{#1}%
					{\ensuremath{#6}}%
					{\ensuremath{{#6}}}}%
				{\IfBooleanTF{#1}%
					{\ensuremath{\@nameuse{@NNClass#2}{#6}}}%
					{\ensuremath{{\@nameuse{@NNClass#2}{#6}}}}}%
			\ifthenelse{\boolean{@NNInListMode}}{}{\protect\stepcounter{@NNCounter\string#3}}%
			\xspace%
		}%
	}%
	\@NNAddEntry{#3}{\string#3}{#4}{#7}{#8}%
}
\DeclareDocumentCommand{\newrelation}{s o m O{0} o m O{} O{}}{%
	\ifdef{#3}%
		{\PackageError{newnotation}{newrelation: Command '\string#3' already defined.}{Misspelled name?}}{}%
	\@NNPrepareEntry{\string#3}%
	\IfNoValueTF{#5}{%
		\newrobustcmd{#3}[#4]{%
			\IfNoValueTF{#2}%
				{\IfBooleanTF{#1}%
					{\ensuremath{\mathrel{#6}}}%
					{\ensuremath{\mathrel{{#6}}}}}%
				{\IfBooleanTF{#1}%
					{\ensuremath{\mathrel{\@nameuse{@NNClass#2}{#6}}}}%
					{\ensuremath{\mathrel{\@nameuse{@NNClass#2}{#6}}}}}%
			\ifthenelse{\boolean{@NNInListMode}}{}{\protect\stepcounter{@NNCounter\string#3}}%
			\xspace%
		}%
	}{%
	\newrobustcmd{#3}[#4][#5]{%
			\IfNoValueTF{#2}%
				{\IfBooleanTF{#1}%
					{\ensuremath{\mathrel{#6}}}%
					{\ensuremath{\mathrel{{#6}}}}}%
				{\IfBooleanTF{#1}%
					{\ensuremath{\mathrel{\@nameuse{@NNClass#2}{#6}}}}%
					{\ensuremath{\mathrel{\@nameuse{@NNClass#2}{#6}}}}}%
			\ifthenelse{\boolean{@NNInListMode}}{}{\protect\stepcounter{@NNCounter\string#3}}%
			\xspace%
		}%
	}%
	\@NNAddEntry{#3}{\string#3}{#4}{#7}{#8}%
}

\DeclareDocumentCommand{\renewnotation}{s o m O{0} o m O{} O{}}{%
	\ifundef{#3}%
		{\PackageError{newnotation}{renewnotation: Command '\string#2' not defined.}%
		{Misspelled name?}}{}%
	\@NNPrepareEntry{\string#3}%
	\IfNoValueTF{#5}{%
		\renewrobustcmd{#3}[#4]{%
			\IfNoValueTF{#2}%
				{\IfBooleanTF{#1}%
					{\ensuremath{#6}}%
					{\ensuremath{{#6}}}}%
				{\IfBooleanTF{#1}%
					{\ensuremath{\@nameuse{@NNClass#2}{#6}}}%
					{\ensuremath{{\@nameuse{@NNClass#2}{#6}}}}}%
			\ifthenelse{\boolean{@NNInListMode}}{}{\protect\stepcounter{@NNCounter\string#3}}%
			\xspace%
		}%
	}{%
	\renewrobustcmd{#3}[#4][#5]{%
			\IfNoValueTF{#2}%
				{\IfBooleanTF{#1}%
					{\ensuremath{#6}}%
					{\ensuremath{{#6}}}}%
				{\IfBooleanTF{#1}%
					{\ensuremath{\@nameuse{@NNClass#2}{#6}}}%
					{\ensuremath{{\@nameuse{@NNClass#2}{#6}}}}}%
			\ifthenelse{\boolean{@NNInListMode}}{}{\protect\stepcounter{@NNCounter\string#3}}%
			\xspace%
		}%
	}%
	\@NNAddEntry{#3}{\string#3}{#4}{#7}{#8}%
}

\makeatother            %

 \usepackage{csquotes}
\usepackage{todonotes}
\usepackage{tcolorbox}
\usepackage{enumerate}
\usepackage{algorithm}
\usepackage{algorithmicx}
\usepackage[noend]{algcompatible}
\usepackage{caption}
\usepackage{subcaption}
\usepackage{cite}
\usepackage{booktabs}

\providebool{hideInternalComments}

\theoremstyle{plain}
\newtheorem{observation}[theorem]{Observation}

\newcommand{\ShrinkSuggestion}{}
\newcommand{\AdjustDefinitionItemizationIndent}{\setlength{\itemindent}{+1.5em}}

\newcommand{\ProofStep}[1]{\medskip\noindent\textsf{#1.}}
\newcommand{\ProofSubstep}[1]{\smallskip\noindent\textit{#1.}}

\newenvironment{new}{}{}

\newcommand{\codelineref}[2][]{%
	\ifthenelse{\equal{#1}{} \OR \equal{\getrefnumber{#1}}{\getrefnumber{#2}}}{Line~\ref{#2}}{Lines \ref{#1}~--~\ref{#2}}%
}

\tikzset{pattern box/.style={draw=gray!50!black,minimum size=1.5em}}
\tikzset{bounded/.style={fill=gray!15!white}}
\newcommand{\nodeBox}[2][]{\raisebox{-0.5em}{\tikz{\node [#1] {\ensuremath{\scriptscriptstyle #2}}}}}
\newcommand{\dotbox}{\nodeBox[pattern box,draw=none]{\scriptstyle\cdots}}
\newcommand{\unbounded}[1]{\nodeBox[pattern box]{#1}}
\newcommand{\bounded}[1]{\nodeBox[pattern box,bounded]{#1}}

\newnotationclass{sf}{\mathsf}
\newnotationclass{bf}{\mathbf}
\newnotationclass{fr}{\mathfrak}
\newnotationclass{cal}{\mathcal}
\newnotationclass{compclass}{\textsf}
\newnotationclass{problem}{\textsc}
\newnotationclass{relation}{\mathtt}
\newnotationclass{attribute}{\text}
\newnotationclass{relationSet}{\mathsf}

\notationnewgroup{general}
\newnotation[cal]{\calA}{A}
\newnotation[cal]{\calB}{B}
\newnotation[cal]{\calC}{C}
\newnotation[cal]{\calD}{D}
\newnotation[cal]{\calK}{K}
\newnotation[cal]{\calM}{M}
\newnotation[cal]{\calL}{L}
\newnotation[cal]{\calN}{N}
\newnotation[cal]{\calU}{U}
\newnotation{\class}{\calC}

\newnotation{\na}{\bar{a}}
\newnotation{\nb}{\bar{b}}
\newnotation{\bb}{\boldsymbol{b}}
\newnotation{\bc}{\boldsymbol{c}}
\newnotation{\bm}{\boldsymbol{m}}
\newnotation{\br}{\boldsymbol{r}}
\newnotation{\bx}{\boldsymbol{x}}
\newnotation{\by}{\boldsymbol{y}}
\newnotation{\bz}{\boldsymbol{z}}
\renewnotation{\phi}{\varphi}
\newnotation{\mydef}{\mathrel{\smash{\stackrel{\scriptscriptstyle{\text{def}}}{=}}}}

\notationnewgroup{general math}
\newnotation{\Nat}{\mathbb{N}}
\newnotation{\Int}{\mathbb{Z}}
\newnotation{\Rat}{\mathbb{Q}}
\newnotation{\powerset}{\mathcal{P}}
\newnotation{\dpowerset}[1][{\le d}]{\powerset_{#1}}
\newnotation{\id}{\mathsf{id}}
\undef\setminus
\newrelation{\setminus}{-}
\newnotation{\bin}{\mathsf{bin}}
\newnotation{\norm}[1]{\left\lVert#1\right\rVert}

\notationnewgroup{complexity theory}
\newnotation{\bigO}{\mathcal{O}}
\newnotation{\poly}{\mathsf{poly}}
\newnotation{\encsize}[1]{\Vert{#1}\Vert}
\newnotation[compclass]{\NP}{NP}
\newnotation[compclass]{\coNP}{coNP}
\newnotation[compclass]{\pitwo}{\ensuremath{\Pi^p_2}}
\newnotation[compclass]{\PSPACE}{PSPACE}
\newnotation[compclass]{\PTime}{PTime}
\newnotation[compclass]{\NPSPACE}{NPSPACE}
\newnotation[compclass]{\EXPTIME}{EXPTIME}
\newnotation[compclass]{\twoEXPTIME}{2EXPTIME}
\newnotation[compclass]{\APSPACE}{APSPACE}

\notationnewgroup{relational database}
\newnotation[sf]{\adom}{adom}
\newnotation[sf]{\dom}{dom}
\newnotation{\domain}{d}
\newnotation[sf]{\var}{var}
\newnotation[sf]{\dvar}{dvar}
\newnotation[sf]{\nvar}{nvar}
\newnotation[sf]{\nulls}{null}
\newnotation[sf]{\ar}{ar}
\newnotation[sf]{\rel}{rel}
\newnotation{\Schemas}[1]{\mathfrak{S}_{#1}}
\newnotation{\schema}{\mathcal{S}}
\newnotation{\globalSchema}{\bar{\mathcal{S}}}
\newnotation{\globalInstance}{\bar{I}}
\newnotation{\famI}{\mathcal{I}}
\newnotation{\fact}{f}
\newnotation{\factA}{g}
\newnotation{\facts}{\mathcal{F}}
\newnotation{\factsB}{\mathcal{G}}
\newnotation{\allfacts}{\mathsf{facts}}
\newnotation{\atoms}{\mathcal{A}}
\newnotation{\atomsB}{\mathcal{B}}
\newnotation{\comparisons}{\mathcal{C}}
\newnotation[fr]{\atomsSetB}{B}
\newnotation{\dtag}[1]{\mathsf{t}_{#1}}
\newnotation{\dtagk}{\dtag{\node}}
\newnotation{\dtagstar}{\dtag{*}}
\newnotation{\undtag}{\mathsf{t}^{-1}}
\newnotation{\local}[1]{{#1}^{@}}
\renewnotation{\deg}[1]{\mathsf{deg}_{#1}}
\newnotation{\cdeg}[1]{\mathsf{cdeg}_{#1}}
\newnotation{\loc}[1][d]{\mathsf{loc}_{#1}}
\newnotation{\loceq}[1][d]{\equiv_{#1}}
\newnotation{\qr}{\mathcal{Q}}
\newnotation{\qrA}{\mathcal{Q}'}
\newrelation{\from}{:\!-}
\newnotation{\body}[1]{\mathsf{body}_{#1}}
\newnotation{\rbody}[1]{\mathsf{rbody}_{#1}}
\newnotation{\cbody}[1]{\mathsf{cbody}_{#1}}
\newnotation{\head}[1]{\mathsf{head}_{#1}}
\newnotation{\bodycomp}[1]{\mathsf{comp}_{#1}}

\newnotation{\At}[3][]{#2\at[#1]#3}
\newnotation{\at}[1][]{{\scriptstyle @}_{\scriptscriptstyle #1}}
\newnotation{\Ifamily}[1][\nw]{(I_\nodek)_{\nodek\in\nw}}
\newnotation{\Qnaive}[1][Q]{Q_{\text{naive}}}
\newnotation{\Tclass}{\mathcal{T}}
\newnotation{\Eclass}{\mathcal{E}}
\newnotation{\Tall}{\Tclass_{\text{all}}}
\newnotation{\Tdf}{\Tclass_{\text{df}}}
\newrelation{\Tos}{\Tclass_{\text{1s}}}
\newnotation{\Tbd}[1][]{\Tclass^{#1}_{\text{bc}}} %
\newnotation{\Tbc}[1][]{\Tclass^{#1}_{\text{bc}}}
\newnotation{\Tbg}[1][]{\Tclass^{#1}_{\text{bg}}}
\newnotation{\Twbd}[1][]{\Tclass^{#1}_{\text{wbc}}}
\newnotation{\Tmid}{\Tclass_{\text{mid}}}
\newnotation{\Eall}{\Eclass_{\text{all}}}
\newnotation{\Ebd}[1][]{\Eclass^{#1}_{\text{bc}}}
\newnotation{\glob}{\natural}
\newnotation{\globR}{R^\glob}
\newnotation{\globS}{S^\glob}
\newnotation{\globT}{T^\glob}

\newnotation{\nvars}{\mathsf{nvar}}
\newnotation{\dvars}{\mathsf{dvar}}
\newnotation{\xvars}{\mathsf{var}}
\newnotation{\svars}{\mathsf{svar}}
\newnotation{\jvars}{\mathsf{jvar}}
\newnotation{\vars}{\mathsf{var}}
\newnotation{\data}{\mathsf{data}}
\newnotation{\context}[1]{\mathsf{cont}_{#1}}

\notationnewgroup{network}
\newnotation{\network}{\mathcal{N}}
\newnotation{\nw}{\network}
\newnotation{\node}{\kappa}
\newnotation{\nodeA}{\lambda}
\newnotation{\nodeB}{\mu}
\newnotation{\nodes}{\calK}
\newnotation{\nodesA}{\calL}
\newnotation{\nodesB}{\calM}
\newnotation{\nodeBseq}{{\boldsymbol\mu}}
\newnotation[bf]{\dist}{I}
\newnotation[bf]{\distH}{H}
\newnotation[bf]{\policy}{P}

\newnotation{\myglobal}{\text{global}}
\newnotation{\mylocal}{\text{local}}

\newnotation{\nodek}{k}
\newnotation{\nodel}{\ell}
\newnotation{\nodem}{m}
\newnotation{\vark}{\kappa}
\newnotation{\varl}{\lambda}
\newnotation{\varm}{\mu}

\notationnewgroup{dependencies}
\newnotation[problem]{\Implication}{Imp}
\newnotation{\Imp}{\Implication}
\newnotation{\gto}{\to}
\newnotation{\lto}{\rightsquigarrow}
\newnotation{\gcond}{\mathsf{G}}
\newnotation{\lcond}{\mathsf{L}}
\newnotation{\bchase}[2]{\mathsf{chase}^{#2}_{#1}}
\newnotation{\chase}[1]{\mathsf{chase}^{}_{#1}}
\newnotation{\Chase}{\mathsf{chase}}
\newnotation{\apply}[1]{\mathsf{apply}^{}_{#1}}
\newnotation[bf]{\chseq}{D}
\newnotation{\base}[1][\chseq]{\mathsf{base}_{#1}}
\newnotation{\ibase}[1][\chseq]{\mathsf{ibase}_{#1}}
\newnotation{\irbase}[1][\chseq]{\mathsf{base}^\circ_{#1}}
\newnotation{\extbase}[1][\chseq]{\mathsf{base}^\infty_{#1}}
\newnotation{\iextbase}[1][\chseq]{\mathsf{ibase}^\infty_{#1}}
\newrelation{\predec}[1][\chseq]{\to_{#1}}
\newrelation{\predecstar}[1][\chseq]{\to^*_{#1}}
\newnotation{\Pred}{\mathsf{Pred}}
\newnotation{\transit}[1][c]{\xrightarrow{#1}}%

\newnotation{\BNI}[1]{\mathcal{B}_{#1}}

\newnotation[sf]{\Gen}{G}
\newnotation[sf]{\Col}{C}
\newnotation{\GenA}{\Gen_1}
\newnotation{\GenB}{\Gen_2}
\newnotation{\ColA}{\Col_1}
\newnotation{\ColB}{\Col_2}

\newnotation[relation]{\Dom}{Dom}

\notationnewgroup{\EXPTIME-hardness proofs}
\newnotation{\Qex}{Q_\exists}
\newnotation{\Qall}{Q_\forall}
\newnotation{\blank}{\texttt{\char32}}
\newnotation{\lmark}{\triangleright}
\newnotation{\rmark}{\triangleleft}
\newnotation{\lmove}{\leftarrow}
\newnotation{\rmove}{\rightarrow}
\newnotation{\stay}{\downarrow}
\newnotation{\SigmaTrans}{\Sigma_{\text{trans}}}
\newnotation{\SigmaAcc}{\Sigma_{\text{acc}}}
\newnotation[relationSet]{\relsId}{Id}
\newnotation[relationSet]{\relsTape}{Tape}
\newnotation[relation]{\relState}{State}
\newnotation{\relStateAll}{\relState_\forall}
\newnotation{\relStateEx}{\relState_\exists}
\newnotation[relation]{\relPos}{Pos}
\newnotation[relation]{\relSym}{Sym}
\newnotation[relation]{\relAlph}{Alph}
\newnotation[relation]{\relHead}{Head}
\newnotation[relation]{\relTrans}{Trans}
\newnotation[relation]{\relSucc}{Succ}
\newnotation[relation]{\relAcc}{Acc}
\newnotation[relation]{\relBit}{Bit}
\newnotation[relation]{\relNext}{Next}
\newnotation[relation]{\relTape}{Tape}
\newnotation[relation]{\relEval}{Eval}
\newnotation[relation]{\RelRange}{Range}
\newnotation[relation]{\RelMessage}{Msg}

\newnotation[relation]{\Emp}{Emp}
\newnotation[relation]{\Sal}{Sal}
\newnotation[relation]{\Addr}{Addr}
\newnotation[relation]{\Lineitem}{Lineitem}
\newnotation[relation]{\CCustomer}{Customer}
\newnotation[relation]{\Orders}{Orders}
\newnotation[attribute]{\name}{name}
\newnotation[attribute]{\mytitle}{title}
\newnotation[attribute]{\salary}{salary}
\newnotation[attribute]{\myaddress}{address}
\newnotation[attribute]{\linekey}{linekey}
\newnotation[attribute]{\orderkey}{orderkey}
\newnotation[attribute]{\custkey}{custkey}
\newnotation[attribute]{\cname}{cname}

\title{Distribution Constraints:\newline The Chase for Distributed Data}
\titlerunning{Distribution Constraints: The Chase for Distributed Data}

\author{Gaetano Geck}{Dortmund University}{gaetano.geck@tu-dortmund.de}{https://orcid.org/0000-0002-8946-9440}{}
\author{Frank Neven}{Hasselt University and transnational University of Limburg}{frank.neven@uhasselt.be}{https://orcid.org/0000-0002-7143-1903}{}
\author{Thomas Schwentick}{Dortmund University}{thomas.schwentick@tu-dortmund.de}{https://orcid.org/0000-0002-1062-922X}{}
\authorrunning{G.~Geck, F.~Neven and T.~Schwentick}

\ccsdesc[500]{Theory of computation~Database constraints theory}
\ccsdesc[500]{Theory of computation~Logic and databases}
\ccsdesc[300]{Information systems~Parallel and distributed DBMSs}

\keywords{tuple-generating dependencies, chase, conjunctive queries, distributed evaluation}
\Copyright{Gaetano Geck, Frank Neven and Thomas Schwentick}

\acknowledgements{%
	We thank Michael Benedikt, Bas Ketsman, Andreas Pieris,
	Phokion Kolaitis, Christopher Spinrath,  Brecht Vandevoort,
	and Thomas Zeume for helpful discussions on various aspects of this work.
}
\begin{document}

\maketitle

\begin{abstract}
This paper introduces a declarative framework to specify and
reason about distributions of data over computing nodes in a
distributed setting. More
specifically, it proposes distribution constraints which are tuple and
equality generating dependencies (tgds and egds) extended with node
variables ranging over computing nodes. In particular, they can
express co-partitioning constraints and constraints about
range-based data distributions by using 
comparison atoms. The main technical
contribution  is the study of the implication problem of
distribution constraints. While implication is undecidable in general,
relevant fragments of so-called data-full constraints are exhibited for
which the corresponding implication problems are complete for
\EXPTIME, \PSPACE and \NP. These results yield bounds on deciding parallel-correctness
for conjunctive queries in the presence of distribution constraints. 
\end{abstract}

\section{Introduction}
\label{sec:Introduction}
Distributed storage and processing of data has been used and studied since the 1970s and became more and more important in the recent past.
One of the most fundamental questions in distributed data management is the following: \emph{how should data be replicated and partitioned over the set of computing nodes?} It is paramount to answer this question well as the placement of data determines the reliability of the system and is furthermore critical for its scalability including the performance of query processing.

On the one hand, despite the importance of this question and decades of research, the placement strategies remained rather simple for a long time: horizontal or vertical fragmentation of relations---or hybrid variants thereof~\cite{DBLP:books/daglib/0029498}. These placement strategies often require a reshuffling of the data for each binary join in the processed query which are commonly based on a range or hash partitioning of the relevant attributes. Recently, however, more elaborated schemes of data placement like co-partitioning, single hypercubes (for multiway-joins) or multiple hypercubes (for skewed data) gained some attention~\cite{DBLP:journals/pvldb/ShuteVSHWROLMECRSA13,DBLP:journals/pvldb/SamwelCHGVYPSTA18,DBLP:conf/sigmod/ZamanianBS15,DBLP:conf/pods/AbiteboulBGA11,DBLP:journals/jacm/BeameKS17,DBLP:conf/pods/KetsmanS17}.

On the other hand, there is a long tradition in studying tuple- and equality-generating dependencies (tgds/egds) as a simple but versatile tool to describe relationships among relational data. The research on these dependencies focuses mainly on the implication\footnote{Does a dependency~$\tau$ always hold if a set~$\Sigma$ of dependencies is satisfied, $\Sigma\models\tau$?} problem. More precisely, since the implication problem in general is undecidable, several fragments have been considered in an attempt to locate the boundaries of decidability and complexity. Commonly, these fragments are defined by syntactical restrictions on the sets of dependencies, like weak acyclicity, weak guardedness, stickiness, wardedness, \dots~\cite{DBLP:journals/tcs/FaginKMP05,DBLP:journals/jair/CaliGK13,DBLP:journals/pvldb/CaliGP10,DBLP:journals/corr/abs-1809-05951}.

\emph{It seems desirable to connect these two strands of research.}
Being able to reason about the placement of data offers database management systems additional optimisation potential, for instance, when it comes to the placement of new data or when the cost of a query execution plan is estimated. In the latter case, a reshuffling phase, which often dominates the processing time, can sometimes be omitted completely because the query at hand is already parallel-correct\footnote{Parallel correctness is a basic notion of distributed query evaluation~\cite{DBLP:journals/jacm/AmelootGKNS17}, also addressed in Section~\ref{sec:appl:dtgd}.} under the current distribution.

\emph{The goal of this paper is to make a first step towards a connection between existing partitioning schemes and well-known reasoning frameworks.}
With this intent, we introduce \emph{distribution constraints}---a variant of tgds/egds that is specifically geared towards distributed data---and study its implication problem. In particular, we identify fragments of distribution constraints by the complexity of the associated implication problem. Although the implication problem is certainly not the only---and, admittedly, not the most innovative---problem related to reasoning about distributed data, it is yet a basic problem that is likely to have connections to other algorithmical questions centering around this topic (like \emph{how to derive a new distribution for the next query, making use of the current distribution?}).

\medskip
\noindent
\textbf{Contributions.}
We start by defining \emph{distribution constraints} as tgds and egds with atoms of the form $R(x,y)\at\node$, in which \node is understood as a node variable with the intended meaning that fact $R(x,y)$ is at node \node. To achieve decidability, we further require that distribution tgds are \emph{data-full}, i.e., only node variables may be quantified existentially. %

We demonstrate that distribution constraints can express several common distribution schemes, incorporating range and hash partitionings~\cite{DBLP:books/daglib/0029498}, 
co-partitionings~\cite{50905,Fushimi:1986:OSS:645913.671448},
hierarchical partitionings (as used in Google's F1~\cite{DBLP:journals/pvldb/ShuteVSHWROLMECRSA13,DBLP:journals/pvldb/SamwelCHGVYPSTA18}), predicate-based reference partitionings~\cite{DBLP:conf/sigmod/ZamanianBS15}, hypercube distributions~\cite{DBLP:conf/pods/AbiteboulBGA11,DBLP:journals/jacm/BeameKS17}, and multi-round communication.
\begin{example}
	As an example, consider the following set of distribution tgds, describing a \enquote{derived horizontal} fragmentation~\cite{DBLP:books/daglib/0029498} of relation~\RelMessage based on the \RelRange-predicate and the message's sender id~$s$:
	\begin{equation*}
		\begin{array}{l}
			\RelRange(\ell,u) \to \RelRange(\ell,u)\at\node, \\
			\RelMessage(s,r) \to \RelMessage(s,r)\at\node, \\
			\RelMessage(s,r)\at\node, \RelRange(\ell,u)\at\nodeA, \ell \le s, s \le u \to \RelMessage(s,r)\at\nodeA
		\end{array}
	\end{equation*}
	 The first two rules enforce that, for every \RelRange- and every \RelMessage-fact, there is a responsible node (indicated by the node variable $\kappa$). The third rule ensures that every \RelMessage-fact can be found at every node whose \RelRange-bounds match the sender id. We remark that the above set of constraints implies the following distribution tgd:
	\begin{equation*}
		\RelMessage(s_1,r), \RelMessage(s_2,r),
		\RelRange(\ell,u),
		\ell \le s_1, s_1 \le u,
		\ell \le s_2, s_2 \le u
		\to
		\RelMessage(s_1,r)\at\node, \RelMessage(s_2,r)\at\node,
	\end{equation*}
	which states that all pairs of messages with the same receiver can be found at a common node if their senders fall in the same range. 
    In other words, if the above set of constraints is satisfied over a distributed instance, then so is the just mentioned dtgd.
	\qed 
\end{example}\label{example:intro}
On the technical level, we show that the implication problem is \EXPTIME-complete for these constraints in general, and we identify classes of distribution constraints where the complexity drops to~\PSPACE or even~\NP and classes where this is not the case. These classes are determined by simple syntactic criteria based on the amount of data associated with node variables.

Since distribution constraints incorporate all \emph{full} tgds (without existential quantification), \EXPTIME-hardness of their implication problem readily follows from an early result by Chandra, Lewis and Makowsky~\cite{DBLP:conf/stoc/ChandraLM81}. However, the latter result relies on the use of relation atoms of arbitrarily high arity, while the \EXPTIME-hardness results in this paper already hold for a fixed schema of maximum arity of~3 (or~2, w.r.t.\  data variables).
The corresponding upper bounds are established by an adaptation of the standard chase procedure~\cite{DBLP:journals/tcs/FaginKMP05,Onet12}.

The fragments studied here are defined depending on, first, the sizes of the node variables' \emph{contexts} (the data variables occuring together with the node variable in some atom) and, second, on the distinction of data-collecting tgds and node-creating tgds (without/with existentially quantified node variable in the head). For a fixed integer~$b$, a node variable has \emph{bounded context} if its context size is at most~$b$. Thanks to the obvious relationship between distribution constraints and standard constraints,  the complexity results in this paper can also be viewed as results on fragments of standard tgds/egds.

\medskip
\noindent
\textbf{Related work.}
There is a rich literature on restrictions of (sets of) tgds that yield a decidable (general and finite) implication problem~\cite{DBLP:books/aw/AbiteboulHV95,Onet12}. We discuss how our distributed constraints relate to classical constraints in Section~\ref{sec:constraints}.  Restricting the use of existential variables in tgds is a common approach to define fragments of tgds that yield a decidable implication problem. Interestingly, the rather simple restriction to data-full dtgds studied here, is orthogonal to prominent examples like weak acyclicity, weak guardedness, stickiness and wardedness~\cite{DBLP:journals/tcs/FaginKMP05,DBLP:journals/jair/CaliGK13,DBLP:journals/pvldb/CaliGP10,DBLP:journals/corr/abs-1809-05951}.

Dependencies with arithmetic comparisons have been used in the context of Data Exchange \cite{AfratiLP08,CateKO13}. However, these papers  mainly study full and weakly acyclic tgds and are thus orthogonal to our framework. There is further work on dependencies with stronger arithmetic constraints, e.g. \cite{MaherS96,BaudinetCW99,DouC13,ArtaleKKRWZ17}.  

Declarative specifications for distributed data have also been studied before.
Notable examples are Webdamlog and the already mentioned Data Exchange setting (which can be seen as a restricted form of distribution constraints with a global and a single local database). We refer to the book \cite{ArenasBLM2014} for a relatively recent overview of Data Exchange.

Our notation $R(x)\at\node$ for distributed atoms resembles that of Webdamlog, $R\at\node(x)$, a dialect of datalog that was designed for distributed data management.\footnote{Annotated atoms have already been used before in Datalog dialects. For instance, in Dedalus~\cite{DBLP:conf/datalog/AlvaroMCHMS10}, where they describe timestamps.} Besides implementing a system~\cite{DBLP:journals/corr/abs-1304-4187,DBLP:conf/sigmod/MoffittSAM15} based on this dialect, the theoretical research on this language has mostly focussed on establishing a hierarchy among some of its fragments in terms of their expressiveness~\cite{DBLP:conf/pods/AbiteboulBGA11}.
Neglecting the notational similarities, there seems to be no overlap between the research on Webdamlog---with its fixpoint evaluation mechanism (which even allows facts to vanish)---and the results on distribution constraints that we present in this paper. Particularly, Webdamlog seems to prohibit existential quantification of node variables and assumes, accordingly, that the number of nodes is explicitly fixed with the input.  Distribution constraints, in contrast, \emph{do allow} existential quantification of node variables, which affects the modeling capabilities and the complexity of the reasoning process.

\medskip
\noindent
\textbf{Organisation of this paper.}
After providing the necessary preliminaries in Section~\ref{sec:Preliminaries}, we formally define distribution constraints in Section~\ref{sec:constraints}, compare them with classical constraints, and give examples of their versatility. %
In Section~\ref{sec:reasoning}, we define the implication problem and extend the standard chase to distribution constraints.  In Section~\ref{sec:complexity}, we address the complexity of the implication problem and, finally, conclude in Section~\ref{sec:Conclusion}.  

\section{Preliminaries}
\label{sec:Preliminaries}

In this section, we fix our notation for the basic concepts of this paper. Specific definitions for our framework are given in Section~\ref{sec:constraints}.

\subsection{Databases and queries}
Let~\dom and~\var be disjoint infinite sets of \emph{data values} and \emph{data variables}, respectively. For simplicity, we do not distinguish between different data types and  assume that \dom is linearly ordered. We denote data variables as usual with $x$, $y$, $z$, \ldots.
A \emph{schema} is a set~\schema of relation symbols, where each relation symbol $R \in \schema$ has some fixed arity $\ar(R)$. We write $\ar(\schema)$ for the \emph{maximum} arity $\ar(R)$ of any $R\in\schema$.
A \emph{relation atom over~\schema} is of the form $R(t_1,\dots,t_k)$ where~$R$ is a relation symbol of arity~$k$ and $t_1,\dots,t_k \in \dom\cup\var$. A relation atom is a \emph{fact}  if $t_1,\dots,t_k \in \dom$.   A \emph{comparison atom} is of the form $t < t'$ or $t \le t'$ with $t,t' \in \dom\cup\var$. The set of data values occuring in a set~\atoms of (relational or comparison) atoms is denoted $\adom(\atoms)$. Similarly, the set of variables in~\atoms is denoted $\vars(\atoms)$.
Instances are finite sets of facts over a given schema~\schema.

A \emph{valuation} for a set $\atoms$ of atoms is a mapping $V: \vars(\atoms) \to \dom$. It \emph{satisfies~$\atoms$ on instance~$I$} if $V(A)\in I$ holds for each relation atom $A\in\atoms$ and $V(t)\theta V(t')$ holds for each comparison atom $t \theta t'$ in $\atoms$. We often denote by $V$ also the extension of $V$ to $\dom$ defined by $V(a)=a$ for every $a\in\dom$.

A \emph{conjunctive query~$Q$} is of the form
\begin{math}
	S(x_1,\dots,x_m) \from R_1(\bz_1),\dots,\allowbreak R_\ell(\bz_\ell),
\end{math}
where the \emph{head} of the query, $\head{Q}=S(x_1,\dots,x_m)$, has a relation atom $S$ not in~$\schema$ and its \emph{body}, $\body{Q}=\{R_1(\bz_1),\dots,R_\ell(\bz_\ell)\}$, is a finite set of relation atoms over~\schema. In the following, all queries are assumed to be \emph{safe}, that is, each variable in the head occurs at least once in some body atom.
If $V$ is a valuation that satisfies $\body{Q}$, we say that~$V$ \emph{derives} fact~$V(\head{Q})$. 
The \emph{result~$Q(I)$} of query~$Q$ on instance~$I$ is the set of all derived facts.

\subsection{Dependencies}

A \emph{tgd} $\sigma $ is of the form $\atoms,\comparisons\to\atoms'$, for sets $\atoms,\atoms'$ of relation atoms and a set $\comparisons$ of comparison atoms with $\vars(\comparisons)\subseteq\vars(\atoms)$. {Here, $\atoms'$ form its \emph{head}, and $\atoms,\comparisons$ its \emph{body}, denoted $\head{\sigma}=\atoms'$ and $\body{\sigma}=\atoms\cup\comparisons$, respectively.} We refer to $\atoms$ by $\rbody{\sigma}$. %
The tgd is called \emph{full} if $\vars(\atoms') \subseteq \vars(\atoms)$. An instance~$I$ \emph{satisfies} a tgd~$\sigma$ if, for every valuation~$V$ of~$\body{\sigma}$ that satisfies $\body{\sigma}$ on $I$, there is an extension~$V'$ onto~\head{\sigma} that satisfies $\head{\sigma}$ on $I$.

An \emph{egd} $\sigma$ is of the form $\atoms,\comparisons\to{x=y}$,  for a set $\atoms$ of relation atoms and a set $\comparisons$ of comparison atoms with $\vars(\comparisons)\cup\{x,y\}\subseteq\vars(\atoms)$.   An instance~$I$ \emph{satisfies} an egd~$\sigma$ if $V(x)=V(y)$ for every valuation~$V$ that satisfies $\body{\sigma}$ on $I$, where $\body{\sigma}$ is defined as for tgds.

Sets of dependencies are satisfied by an instance if each dependency in the set is satisfied. Satisfaction of a single dependency~$\sigma$ or a set~$\Sigma$ of dependencies by some instance~$I$ is denoted $I \models \sigma$ and $I \models \Sigma$, respectively.

A dependency~$\tau$ is \emph{implied} by a set~$\Sigma$ of dependencies, denoted $\Sigma \models \tau$, if $I \models \Sigma$ implies $I \models \tau$ for every instance~$I$. For more precise statements, we can mention the actual domain in our notation. For example, we write $I \models_\Nat \tau$ if implication holds for all (finite) instances over \Nat.

We use the terms \emph{dependencies} and \emph{constraints} interchangeably. 

\subsection{Distributed Databases}

We model a \emph{network} of database servers as a finite set $\nw$ of \emph{nodes}  and we denote its \emph{size} by $|\nw|$. We usually denote nodes by $\nodek$ and $\nodel$.
A \emph{distributed instance} $D=(G,\dist)$ consists of a \emph{global instance}~$G$ and a family $\dist=\Ifamily$ of \emph{local instances}, one for each node of $\nw$, such that $\bigcup I_\nodek \subseteq G$.
We denote $G$ by $\myglobal(D)$ and $\Ifamily$ by $\mylocal(D)$.

We note that distributions allow redundant placement of facts, which is often desirable. Furthermore, it is not necessary to place all facts of the global instance on some node. A fact~\fact is \emph{skipped}\footnote{We note that allowing skipped facts makes the framework more flexible. They can be disallowed by simple distribution constraints, as discussed in Subsection~\ref{subsec:globallocal}.} by $D$ if $\fact \in \myglobal(D)$ but $\fact$ does not occur  in $\mylocal(D)$.

We write $\At[D]{f}{\nodek}$ to denote that a fact $f$ occurs at some node $\nodek$, that is, $f\in I_\nodek$. We drop $D$ if it is clear from the context. We call $\At[D]{f}{\nodek}$ a \emph{distributed fact}.
Sometimes we say that a set of facts \emph{meet} in $D$ when they all occur in the same local instance.

\begin{example}
  \label{ex:RepresentationOfDistributions}%
 Consider a network $\nw=\{1,2\}$ of size 2 and a distributed instance $D=(G,\{I_1,I_2\})$ with $G = \{R(a,b),S(b),S(c),S(d)\}$,
 $I_1=\{R(a,b),S(b)\}$ and $I_2=\{S(b),S(c)\}$. Then fact $S(d)$ is skipped by $D$. The instance $D$ can also be represented by the distributed
  facts $\{
 R(a,b),S(b),S(c),S(d),
 \At{R(a,b)}{1}, \At{S(b)}{1}, \At{S(b)}{2}, \At{S(c)}{2}\}$. \qed
\end{example}

\subsection{Parallel-correctness}
Building on the computation model of massively parallel communication (MPC)~\cite{DBLP:journals/jacm/BeameKS17}, the naive evaluation of a conjunctive query~$Q$ over a distributed instance $D$ %
evaluates $Q$ separately for each local instance in $\mylocal(D)$.
For $\mylocal(D)=\Ifamily$, we write $\Qnaive(D)$ for $\bigcup_{\nodek\in\nw} Q(I_\nodek)$. 
Following~\cite{DBLP:journals/jacm/AmelootGKNS17}, we say that a query~$Q$ is \emph{parallel-correct} on $D$, if the naive evaluation produces the correct result, i.e., if $\Qnaive(D)=Q(\myglobal(D))$. 

\section{Distribution constraints}
\label{sec:constraints}

We first introduce our framework for distribution
constraints and afterwards give examples for its use. 

\subsection{Definition}\label{sec:def:dtgd}
 Let \nvar be an
infinite set of {\em node variables} disjoint from \dom and \var. 
A \emph{distributed atom} $\At{A}{\vark}$ consists of a relation atom $A$ and a \emph{node variable} $\vark$ in \nvar.
Recall that we refer to the variables of $A$ as $\emph{data
variables}$. For a set of (distributed) atoms $\calA$, we denote by
$\nvar(\calA)$ the set of node variables occurring in atoms in
$\calA$. For a set $\calA$ of relation atoms and a node variable $\vark$,
$\At{\calA}{\vark}$ denotes the set $\{\At{A}{\vark} \mid
A\in\calA\}$.

\emph{Distribution tgds (dtgds)} are defined just as tgds but they can
additionally have distributed atoms in their body and their head.
\emph{Distribution egds (degds)} are defined just as egds but can have
distributed atoms in their body. We do not allow node variables in
comparison atoms ({but we do allow them in the equality atom of a head in the case of degds}). A degd $\atoms,\comparisons \to A'$ is \emph{node-identifying} if the equality atom~$A'$ refers to node variables only and \emph{value-identifying} if,  instead,~$A'$ refers to data variables
	only. We do \emph{not} consider equality atoms where a
node variable is identified with a data variable.
We are particularly interested in \emph{data-full} dtgds, for which  the data variables in the head all occur in the body.

By $\Tall$ we denote the class of \emph{all} dtgds and by $\Tdf$  the
class of {data-full} dtgds. By $\Eall$ we denote the class of \emph{all} degds.

Satisfaction of dtgds and degds is defined in the obvious way with
generalised  valuations that may additionally map node variables to nodes.
For a distributed atom $A'=\At{A}{\kappa}$, we write $V(A')\in D$, if $V(A)\at{}V(\kappa)$ is a distributed fact of $D$.
For a relation atom $A$, we write $V(A)\in D$ if $V(A)\in \myglobal(D)$. 

\begin{example}
	Given a schema~\schema with binary relation symbols~$R$ and~$S$, the following dtgd 
	\begin{math}
		\sigma = R(x,y), S(x,y) \to \At{R(x,y)}{\node}, \At{S(x,y)}{\node}
	\end{math}
	is satisfied on a distributed instance $D$ if, whenever $\myglobal(D)$ contains two facts $R(a,b)$ and~$S(a,b)$, for arbitrary data values $a,b \in \dom$, they meet in some local instance. \qed %
\end{example}

Below, in~Section~\ref{sec:more:exam:dtgd}, we illustrate how distribution constraints can model %
global, local and global-to-local constraints.

\begin{example} 
	The dtgd $\At{E(x,y)}{\kappa},\At{E(y,z)}{\kappa},E(z,x) \to \At{E(z,x)}{\kappa}$
	stipulates that every computing node has `complete' information w.r.t.\ open triangles on a binary relation~$E$. 
	That is, whenever a node contains two legs of a triangle, it also contains the closing leg if it exists in the global database. \qed
\end{example}

Clearly, the differentiation between node and data variables in
dtgds/degds can be seen as just syntactic sugar for standard
relational schemas. The above restrictions (at most one node variable,
at a fixed position, data-fullness) can then be seen as restrictions
of classical constraints. In this sense, a dtgd like $R(x)\at\node, S(x)\at\nodeB \to T(x)\at\node$ could be rewritten into a standard tgd of the form $R(\node,x), S(\nodeB,x) \to T(\node,x)$. The restriction to data-full dtgds thus translates to the restriction of existential quantification to these first attributes.

However, as the following example illustrates, our restriction to existential quantification of node
variables does not translate into any of the restricted fragments
with low complexity, which we
are aware of. 
\begin{example}
  Let $\Sigma$ consist of a node-creating
  dtgd $R(x)\at\node \to T(x)\at\nodeB$ and a data-collecting dtgd
  $T(x)\at\node, T(y)\at\node, T(z)\at\nodeB, T(w)\at\nodeB \to
  U(x,y,z,w)\at\node$. The corresponding set of standard tgds%
  \begin{eqnarray*}
          R(\node,x) & \to & T(\nodeB,x),
    \\
          T(\node,x), T(\node,y), T(\nodeB,z), T(\nodeB,w) & \to & U(\node,x,y,z,w),
  \end{eqnarray*}
  is neither sticky nor weakly guarded nor warded.\footnote{The
    set~$\Sigma'$ is not sticky because the marked variable~\nodeB
    occurs more than once in~$\tau'$. It is not (weakly) guarded
    because a single atom cannot contain both variables~\node
    and~\nodeB that occur in affected positions of~$\tau'$. Finally,
    it is not warded because the dangerous variable~\node appears in
    more than one atom in the body of~$\tau'$.} The set~$\Sigma$ has,
  however, bounded context (and its associated implication problem is
  shown to be in~\NP in Section~\ref{sec:complexity}). 
\end{example}
Furthermore, the set consisting  of $R(x,y)\at\node \to S(x,y)\at\nodeB$ and
    $S(x,x)\at\node \to R(x,x)\at\nodeB$ is not weakly acyclic but
    data-full with
    bounded context.

\subsection{Examples of distribution constraints}
\label{sec:more:exam:dtgd}

In the following, we provide examples illustrating the versatility of distribution constraints.
We begin with an examination of certain uses of distributed atoms. In principle, distributed atoms can be used in the body and in the head of constraints, referring to multiple node variables. Some more restricted uses seem particularly useful however.%

We use the schema
\begin{math}
  \ShrinkSuggestion
   \{
     \Emp(\name,\mytitle),\ 
     \Sal(\mytitle,\allowbreak \salary),\ 
     \Addr(\name,\myaddress)
  \}
\end{math}
as a running example for the remainder of this section.

\subsubsection{Global dtgds and degds}

We call distribution constraints \emph{global} if they do not contain any distributed atom.
These constraints refer to the global instance of a distributed database only---irrespective of the local databases.
Formally, a dtgd (resp., degd) $\sigma$ is a \emph{global constraint} if $\sigma$ is a tgd (resp., egd).

\begin{example}\label{ex:schema}
  The following constraints are examples of a global dtgd and a
  global degd:
  \begin{math}
    \ShrinkSuggestion
    \Emp(n,t) \to \Sal(t,s),
  \end{math}
  and
  \begin{math}
    \ShrinkSuggestion
    \Sal(t,s), \Sal(t,s')  \to  s=s'.
  \end{math}
  Together they specify that every employee has a unique salary. \qed
\end{example}

\subsubsection{Local dtgds and degds}
Distribution constraints where every relation atom is a distributed atom and where all these atoms refer to the same node variable are called \emph{local}. These constraints specify conditions that hold on \emph{every} local instance, viewed on its own---irrespective of the global instance or other local instances.
\begin{example}
  The following is a local dtgd expressing that whenever a fact 
  $\Emp(a,b)$ occurs at node $\nodek$ there is a fact
  $\Sal(b,c)$, for some element $c$ in \dom, that occurs at node $\nodek$ as well:
  \begin{math}
    \ShrinkSuggestion
    \At{\Emp(x,y)}{\kappa}\to  \At{\Sal(y,z)}{\kappa}.
  \end{math}
  The following local value identifying degd expresses that, relative to each node, each employee (name) has a unique address.
  \begin{math}
    \ShrinkSuggestion
    \At{\Addr(x,y)}{\kappa}, \At{\Addr(x,y')}{\kappa}\to  y=y'.
  \end{math}
  \qed
\end{example}

\subsubsection{Global-Local dtgds}\label{subsec:globallocal}
Lastly, we call a dtgd \emph{global-local} if none of its body atoms is distributed while all its heads atoms are distributed and refer to the same node variable.

\begin{example}
  The global-local constraint
  \begin{math}
    \ShrinkSuggestion
    \Emp(x,y), \Sal(y,z) \to  \At{\Emp(x,y)}{\kappa}, \At{\Sal(y,z)}{\kappa}
  \end{math}
  expresses that if there is an $\Emp$-fact and a $\Sal$-fact with the same $\mytitle$-attribute then these facts meet at some node. This means that the join condition between $\Emp$ and $\Sal$ induced by the schema is maintained in the horizontal decomposition of the global database.\qed
\end{example}
                    
Global-local constraints can also express that the database has no
skipped facts, i.e., facts $\fact \in \myglobal(D)$ with $\fact \notin
\mylocal(D)$. To this end, for each relation symbol $R$ a global-local
constraint $R(x_1,\ldots,x_{\ar(R)})\allowbreak \to
R(x_1,\ldots,x_{\ar(R)})\at\vark$ can be added. Indeed, it is this
ability of distributed constraints to disallow skipped facts that made us allow them in
first place. For a
schema \schema, we denote by $\calU(\schema)$ the set of all
global-local constraints that express that there are no skipped
facts.%
                    
By symmetry also local-global constraints can be defined. An example would be the dtgd
$\At{\Emp(x,y)}{\kappa}, \At{\Sal(y,z)}{\kappa} \to \Addr(x,z')$ (even though for this particular schema, the constraint is rather contrived).
Nevertheless, local-global constraints allow to state explicitly that
every local fact is also a global fact, by stating
$\At{R(x_1,\ldots,x_m)}{\kappa} \to R(x_1,\ldots,x_m)$, for every
relation $R$ (with $m=\ar(R)$).

\subsection{Applications of distribution constraints}
\label{sec:appl:dtgd}

We give some applications of distribution constraints like 
defining range, hash and co-partitionings and testing for  parallel-correctness. In the appendix (\ref{sec:more:appl:dtgd}), we illustrate how hypercube distributions can be incorporated and discuss query answering and multi-round query evaluation.
The paper~\cite{DBLP:conf/icdt/NevenSSV19}  further explores the use of distribution constraints to model distributed evaluation strategies for Datalog in the context of parallel-correctness and parallel-boundedness in the multi-round MPC model.%

\subsubsection{Range and hash partitioning}

Distribution constraints can easily incorporate the commonly used range and hash partitionings~(see for example \cite{DBLP:books/daglib/0029498,DBLP:books/oreilly/Kleppmann2014}). Example~\ref{example:intro} already illustrates range partitionings.

The following two distribution constraints define a hash partitioning of the relation $\Emp(\text{name},\text{dept})$ on the attribute department: 
\begin{eqnarray*}
{\Emp(n,d)}\to  \At{\Emp(n,d)}{\kappa}\\
\At{\Emp(n,d)}{\kappa},{\Emp(n',d)} \to \At{\Emp(n',d)}{\kappa}
\end{eqnarray*}
The first rule enforces that every Emp-tuple occurs at a node while the second rule ensures that Emp-tuples within the same department are placed together. The above approach where hash functions are implicit should be contrasted with the modeling of Hypercube distributions, discussed in the appendix (\ref{app:hc}), where hash functions are made explicit.

\subsubsection{Co-partitioning}

\begingroup
A popular way to avoid expensive remote join operations---already used in early parallel systems---is to co-partition tables on their join key~\cite{50905,Fushimi:1986:OSS:645913.671448}.
Generalizations of the latter technique where co-partitioning is determined by more complex join predicates have been shown to be effective in modern systems as well \cite{DBLP:journals/pvldb/SamwelCHGVYPSTA18,DBLP:journals/pvldb/ShuteVSHWROLMECRSA13,DBLP:conf/sigmod/ZamanianBS15,DBLP:conf/icde/RodigerMURK014}. 

Consider, for instance, the following (simplified) relations from the TPC-H schema~\cite{tpc-h}: $\Lineitem(\linekey,\orderkey)$, $\Orders(\orderkey,\custkey)$, and $\CCustomer(\custkey,\cname)$.

Zamanian, Binnig, and Salama~\cite{DBLP:conf/sigmod/ZamanianBS15} exemplify the following co-partitioning scheme: \Lineitem is hash-partitioned by linekey, \Orders tuples are co-partitioned with \Lineitem tuples with the same orderkey, and \CCustomer tuples are co-partitioned with \Orders tuples with the same custkey.
As a consequence, the join
$\text{\Lineitem} \bowtie \text{\Orders} \bowtie \text{\CCustomer}$ can be evaluated without expensive remote joins.
We note that the work in \cite{DBLP:conf/sigmod/ZamanianBS15} is by no means restricted to single-round or communication-free evaluation of queries. Knowledge of co-location of tuples is used to rewrite query plans and to determine those parts that can be evaluated without additional reshuffling. Partitionings are also not considered to be static but should adapt over time to changes in workload and the data~(e.g., \cite{DBLP:journals/pvldb/LuSJM17,DBLP:journals/pvldb/SerafiniTEPAS16}). 

\begin{example}
	\newnotation{\xc}{c}
	\newnotation{\xl}{\ell}
	\newnotation{\xn}{n}
	\newnotation{\xo}{o}
	\newnotation{\xoB}{o'}
	\rm \label{example:lineitem}
	The following distribution constraints define the co-partitioning scheme mentioned above:
	\begin{eqnarray}
		\Lineitem(\xl,\xo) \to \At{\Lineitem(\xl,\xo)}{\kappa} \label{meq1}\\
		\Orders(\xo,\xc)\to \At{\Orders(\xo,\xc)}{\kappa}\label{meq11}\\
		{\CCustomer(\xc,\xn)}\to\At{\CCustomer(\xc,\xn)}{\kappa} \label{meq111}\\
		\Lineitem(\xl,\xo), \At{\Lineitem(\xl,\xoB)}{\kappa} \to \At{\Lineitem(\xl,\xo)}{\kappa}\label{meq2}\\
		\At{\Lineitem(\xl,\xo)}{\kappa}, \Orders(\xo,\xc)\to \At{\Orders(\xo,\xc)}{\kappa}\label{meq3}\\
		\At{\Orders(\xo,\xc)}{\kappa},{\CCustomer(\xc,\xn)}\to\At{\CCustomer(\xc,\xn)}{\kappa} \label{meq4}
	\end{eqnarray}
	Basically, Constraint~\eqref{meq1} expresses that every \Lineitem fact in the global database occurs at some node; similarly for Constraints~\eqref{meq11} and~\eqref{meq111}. Constraint~\eqref{meq2} then expresses that \Lineitem facts are hashed on the first attribute: for every item~$\ell$ with order~$o'$ stored on some server, every other order of that item is stored there too. Constraint~\eqref{meq3} expresses that $\Orders(\xo,\xc)$ facts are co-located with $\Lineitem(\xl,\xo)$ facts, while Constraint~\eqref{meq4} expresses that $\CCustomer(\xc,\xn)$ facts are co-located with $\Orders(\xo,\xc)$ facts. All together the distribution constraints imply that the join condition between the three relations is maintained in the horizontal decomposition of the global database.  \hfill \qed
\end{example}
\endgroup

\subsubsection{Hierarchical partitioning schemes}
\label{sec:hierar}
\newnotation[relation]{\Customer}{Cust}
\newnotation[relation]{\Supplier}{Supp}
\newnotation[relation]{\Campaign}{Camp}
\newnotation[relation]{\AdGroup}{AdGrp}
\newnotation[text]{\custId}{custId}
\newnotation[text]{\supId}{supId}
\newnotation[text]{\campId}{campId}
\newnotation[text]{\nation}{natKey}
Recent database systems like Google's F1~\cite{DBLP:journals/pvldb/SamwelCHGVYPSTA18,DBLP:journals/pvldb/ShuteVSHWROLMECRSA13} use hierarchical partitioning schemes to provide performance while ensuring consistency under updates. Hierarchical partitioning is a variant of the \emph{co-partitioning} approach~\cite{DBLP:conf/sigmod/SundarmurthyKN18}, introduced as \emph{predicate-based reference partitioning}~\cite{DBLP:conf/sigmod/ZamanianBS15}. This approach allows to formulate a hashing condition for a relation~$S$, given that a relation~$R$ is already distributed, in the following style: first, every $S$-fact has to be distributed, and second, if an $S$-fact joins with an $R$-fact on a predefined set of attributes, then the $S$-fact is distributed to every node where such an $R$-fact exists.

This is easily modeled by the following dtgds:
\begin{eqnarray*}
	S(\bz) \to S(\bz)\at{\vark}, \\
	R(\by)\at{\vark}, S(\bz) \to S(\bz)\at{\vark},
\end{eqnarray*}
where variables $\by$ and~$\bz$ share some common variables $x_1,\dots,x_n$ representing the join predicate.
Notice that Example~\ref{example:lineitem} follows this scheme. Another example, illustrating Google's AdWord scenario, is given in Appendix~\ref{app:HierarchicalPartitioning}.

\subsubsection{Parallel-Correctness}

We show that parallel-cor\-rect\-ness of a conjunctive query can be captured by a dtgd but not always by a data-full one. 

\begin{example}
	Let $Q = H(n,s) \gets \Emp(n,t),\Sal(t,s)$
	be a conjunctive query. Then, $Q$ is parallel-correct on a distributed instance $D$ if every fact from $Q(\myglobal(D))$ is derived at some node (due to the monotonicity of CQs, we do not need to check the converse statement). This can be expressed by the dtgd
	\begin{math}
		\ShrinkSuggestion
		\Emp(x,y), \Sal(y,z) \to  \At{\Emp(x,y')}{\kappa}, \At{\Sal(y',z)}{\kappa}. 
	\end{math}
	We note that in this dtgd $\kappa$ and $y'$ occur only in the head.
	So, this dtgd is \emph{not} data-full. \qed    
\end{example}

\section{Reasoning}
\label{sec:reasoning}

We consider the implication problem for distribution constraints in Section~\ref{sec:implication}, and adapt the chase to degds and data-full dtgds in Section~\ref{sec:chase}, as a means to solve it.

\subsection{The implication problem}
\label{sec:implication}

We stress that the implied dependency $\tau$ in the definition below, is not required to belong to the class $\class$. That is, $\tau$ can be an arbitrary distribution constraint.
\begin{definition}
  The \emph{implication problem~$\Implication(\class,\domain)$}, parameterised by a class~\class of dependencies and a domain \domain asks, for a finite set~$\Sigma$ from~\class and a single distribution constraint~$\tau$, whether $\Sigma\models_\domain\tau$. Possible choices for \domain are \Nat, \Int and \Rat. If the choice of \domain does not matter or is clear from the context, we also write ~$\Implication(\class)$.
  For each $\alpha\ge 1$, we denote by $\Implication_\alpha(\class,\domain)$ the restriction of $\Implication(\class,\domain)$ to inputs $(\Sigma,\tau)$ in which the arity of each relation symbol (with respect to data) is at most $\alpha$.  
\end{definition}
Since every tgd is a dtgd and the implication problem for   tgds without comparison atoms is undecidable, we instantly get the following.
\begin{observation}[$\!\!$\cite{DBLP:conf/icalp/BeeriV81,DBLP:conf/stoc/ChandraLM81}]
  $\Implication(\Tall)$ is undecidable.
\end{observation}
To facilitate automatic reasoning, it thus makes sense to consider restricted
kinds of constraints. An immediate observation is that most of the
examples in Section~\ref{sec:constraints} only use data-full distribution constraints. In the remainder of this paper, we therefore restrict our attention to this class.
\begin{remark}\label{rem:single-head}
A full tgd $\sigma=\atoms,\comparisons \to
\{A'_1,\dots,\allowbreak A'_p\}$ %
can be
transformed into an equivalent set of tgds $\{\atoms,\comparisons \to
A'_1,\dots,\atoms,\comparisons \to A'_p\}$ with a singleton head.%
\footnote{This observation is used also used the context of normalised schema mappings \cite{DBLP:conf/pods/FaginPKT04,DBLP:journals/vldb/GottlobPS11}.}%
In particular, this applies to dtgds \emph{without existential quantification}. Similarly, data-full dtgds \emph{with existentially quantified node variables} can be decomposed into  data-full dtgds with  at most one node variable in their head. We thus assume w.l.o.g.\ in \emph{upper bound proofs} that all dtgds in $\Sigma$ are decomposed in this fashion. \qed
\end{remark}
Since data-full dtgds have at most one node variable in their head, we can distinguish three kinds of data-full dtgds:
\begin{itemize}
	\item  \emph{node-creating dtgds} like $R(x,y) \to R(x,y)\at\vark$, in which the one node variable in their head is  existentially quantified; 
	\item \emph{data-collecting dtgds} like $S(x)\at\vark, T(x)\at\varl \to T(x)\at\vark$, which are dtgds that have one distributed head atom without existential quantification (they collect facts in a local node); and,
	\item \emph{global dtgds} like $R(x,y) \to U(x)$ or $R(x,y)\at\vark, T(x)\at\vark \to S(y)$ that have one head atom without node variable (they contribute
	global facts).\footnote{The term \emph{global dtgd} thus represents a superset of global and local-global dtgds from Section~\ref{sec:more:exam:dtgd}.}%
\end{itemize}
For brevity, we sometimes refer to \emph{generating} and \emph{collecting} dtgds.

We call the unique node variable that occurs in the head of a node-creating or data-collecting dtgd $\sigma$ the \emph{head variable} of $\sigma$.

\subsection{The chase for distribution constraints}
\label{sec:chase}

The classical way for deciding implication of tgds and egds builds on
the \emph{chase} procedure. In the following, we adapt the chase from \cite{DBLP:journals/tcs/FaginKMP05} for distribution 
constraints from $\Tdf$ and $\Eall$.
\begin{definition}[chase step]
  A  dtgd $\sigma$ is  \emph{applicable} to a distributed instance $D$ with  a  valuation~$W$  if $W$ satisfies $\body{\sigma}$ on $D$ and  there exists no valuation~$W'$ for~$\sigma$ identical to~$W$ on~\body{\sigma} such that $W'(\head{\sigma}) \subseteq D$. Furthermore, if $\sigma$ is node-creating then $W(\kappa)$ must be a node $k$ not occurring in $D$, where $\kappa$ is the head variable of~$\sigma$. 
 The \emph{result} $\Chase(c,D)$ of applying $c=(\sigma,W)$ on $D$ is the distributed instance
	$ D' = D \cup W(\head{\sigma})$ and we write $D\transit D'$. 
      \end{definition}%
	  For instance, if a distributed database~$D$ consists of facts $R(a)\at{1}$ and~$S(a,b)\at{2}$, then dependency $\sigma = R(x)\at\vark, S(x,y)\at\varl \to R(x)\at\varm,S(x,y)\at\varm$ is applicable to~$D$, as witnessed by the valuation~$W$ where $W(x,y) = (a,b)$ and $W(\vark,\varl,\varm) = (1,2,3)$. Application leads to $D' = D \cup \{R(a)\at{3}, S(a,b)\at{3}\}$.

If  $\sigma$ is node-creating, we say that
          the chase step \emph{generates} the new node $W(\vark)$
          with an initial set $W(\head{\sigma})$ of facts. If $\sigma$ is data-collecting, we say that
          the chase step \emph{collects} the facts from
          $W(\head{\sigma})$ in node $W(\vark)$.  If $\sigma$ is global, we say that the chase
          step \emph{targets} the global database.  The chase step
          \emph{contributes} the set~$W(\head{\sigma})$ of global
		  facts. 

\begin{definition}[chase sequence]
Let~$\Sigma$ be a set of
distribution constraints and $D$ a distributed instance. A \emph{chase sequence} for $D$ with $\Sigma$ is a sequence $\chseq=D_0,D_1,\ldots$ of distributed instances with $D_0=D$ and $D_i\transit[(\sigma_i,W_i)] D_{i+1}$, for every $i\ge 0$ and some $\sigma_i\in\Sigma$ and valuation~$W_i$. We write  $\Chase(\chseq)$ for the final instance of $\chseq$, if \chseq is finite.
        
        A chase sequence   $\chseq$ \emph{fails}, if there is a degd $\sigma\in\Sigma$ with a head $t=t'$ and a valuation $W$, such that $W$ satisfies  $\body{\sigma}$ on $\Chase(\chseq)$ and $W(t)\not=W(t')$. It        is \emph{successful}, if it is finite, does not fail and $\Chase(\chseq)$ has no applicable chase step. 
\end{definition}

The following easy observation is crucial for our results.
\begin{proposition}\label{prop:chase-terminates}
  For each distributed database $D$ and each set $\Sigma$ of constraints from $\Tdf$ and $\Eall$, there are no infinite chase sequences for $D$. %
\end{proposition}

For classical tgds and egds, to test $\Sigma\models\tau$, the chase is
basically applied to a \enquote{canonical database} $V(\rbody{\tau})$,
for some one-one valuation $V$. Due to comparison atoms, this does not
suffice in our setting.  Instead, we  consider a set of canonical
databases, which allows for all possible linear orders on the
variables of  $\body{\tau}$. It depends on the general domain $\domain$ which
we allow to be one of $\Nat$, $\Int$, $\Rat$.
More precisely, it is defined over a set $\dom(\Sigma,\tau,\domain)$ of data values that
contains all constants of $\Sigma$ and $\tau$ and, between each pair of successive
constants all intermediate values from $\domain$ or as many
intermediate values as there are variables in $\body{\tau}$. The set of
canonical databases then consists  of all databases of the form
$V(\rbody{\tau})$ for valuations whose range is in
$\dom(\Sigma,\tau,\domain)$.

Towards a formal definition, \mbox{$c_1<\cdots<c_\ell$} denote the
constants in  $\Sigma\cup\{\tau\}$ and let $m$ be the number of data variables
in $\body{\tau}$. If $\domain=\Rat$ then $\dom(\Sigma,\tau,\domain)$ consists of $c_1,\ldots,c_\ell$, all values $c_1-m,\ldots,c_1-1$, all values $c_\ell+1,\ldots,c_\ell+m$ and all values of the form $c_i+\frac{j}{m+1}(c_{i+1}-c_i)$, for $i\in\{1,\ldots,\ell-1\}$ and $j\in\{1,\ldots,m\}$. If  $\domain=\Int$, it consists of $c_1,\ldots,c_\ell$, all values $c_1-m,\ldots,c_1-1$, all values $c_\ell+1,\ldots,c_\ell+m$ and all values of the form $c_i+j$, for $i\in\{1,\ldots,\ell-1\}$ and $j\in\{1,\ldots,m\}$ with $c_i+j<c_{i+1}$.  If  $\domain=\Nat$, it is defined as for $\domain=\Int$ with the additional constraint that elements of the form $c_1-j$ must be non-negative. 

By $\calD(\Sigma,\tau,\domain)$ we denote the set of all distributed databases $V(\rbody{\tau})$, for which $V$ maps data variables to values in  $\dom(\Sigma,\tau,\domain)$ and node variables one-one to an initial segment of (a disjoint copy of) the natural numbers.

The following result shows that, to decide implication, it suffices to apply the chase to all databases in $\calD(\Sigma,\tau,\domain)$.
\begin{proposition}\label{prop:chase-correct}
   Let $\Sigma\cup\{\tau\}$ be a set distribution 
constraints from $\Tdf$ and $\Eall$ and let $\domain$ be one of $\Nat$, $\Int$, $\Rat$. Then the following statements are equivalent.
\begin{enumerate}[(1)]
\item $\Sigma\models_\domain \tau$.
\item For every database $D=V(\rbody{\tau})$ in $\calD(\Sigma,\tau,\domain)$ and every
  successful chase sequence $\chseq $ for $D$ with $\Sigma$, there is an extension $V'$ of $V$ that satisfies $\head{\tau}$ on $\Chase(\chseq)$.
  \item For every database $D=V(\rbody{\tau})$ in $\calD(\Sigma,\tau,\domain)$ there exists a chase sequence $\chseq $ for $D$ with $\Sigma$, that fails or for which there is an extension $V'$ of $V$ that satisfies $\head{\tau}$ on $\Chase(\chseq)$. 
\end{enumerate}
\end{proposition}
The straightforward proof is given in the appendix.

\section{Complexity}
\label{sec:complexity}

In this section, we study the complexity of the implication problem
for data-full  dtgds (and arbitrary degds). In general, this problem
turns out to be  in \EXPTIME, in fact as \EXPTIME-complete. We then
study restrictions of dtgds and degds that lower the complexity of the
implication problem. In fact, we identify fragments whose implication
problems are \pitwo-complete (\NP-complete without comparison atoms) or
\PSPACE-complete. To wrap up the picture, we finally identify fragments that already yield  \EXPTIME-hardness.

For most practical cases, the relevant complexity is $\pitwo$ or even
$\NP$: the former is the case, e.g., if the database schema (or at
least its arity) is fixed and if the number of atoms (or at least the
number of variables) is bounded by some a-priori constant. The latter
is the case if, additionally, there are no comparison atoms. In particular, the
\enquote{natural} generalisations of our examples have at most $\pitwo$
(or \NP) complexity.

  The first result of this section states that  $\Implication(\Tdf)$
  is \EXPTIME-complete. The upper bound is very simple but shows that
  the problem is not harder than implication of full (non-distributed)
  tgds. On the other hand,  in the distributed setting,
  \EXPTIME-hardness already holds for fixed schemas with small arity,
  whereas this problem is easily seen to be in~\pitwo for
  (non-distributed) full
  dependencies.%

\begin{theorem}
	\label{thm:ImplicationForDataFullDTGDisEXPTIMEComplete}%
	$\Implication(\Tdf\cup \Eall)$ is  $\EXPTIME$-complete. The lower bound already holds for a fixed schema of arity 2.
\end{theorem}
\begin{proof}[Proof sketch]
  The lower bound follows from Theorem~\ref{prop:lowerbounds}, which is shown in Subsection~\ref{sec:lowerbounds} and offers a collection of types of distributed constraints that make the implication problem \EXPTIME-hard.
  
  The upper bound uses Proposition~\ref{prop:chase-correct}. Since
  $\dom(\Sigma,\tau,\domain)$ has polynomial size in
  $|\Sigma|+|\tau|$, the set $\calD(\Sigma,\tau,\domain)$ contains
  only exponentially many databases, from which the chase needs to start.
 Furthermore, each constraint in $\Sigma$ can
 be applied at most an exponential number of times, since there are
 only exponentially many different valuations, and
 each of them can fire at most once. Therefore, each chase sequence is
 of at most exponential length.
\end{proof}  

Intuitively, \EXPTIME-hardness for the class~\Tdf of data-full dtgds
(and even without degds) follows from the need to keep track of an
exponential number of nodes, as can be seen from the proof of the
lower bound of Theorem~\ref{prop:lowerbounds}. 

In the following two subsections, we turn to  restricted classes with
lower complexity\footnote{This statement holds under the common assumption that $\pitwo$
  and \PSPACE are smaller than \EXPTIME.} for the implication problem.  The fragments that we study are not
motivated from practical considerations (since there we already have
\enquote{low} complexity).
They were rather obtained by considering syntactic properties
under which the chase behaves better than in general.

First of all, these fragments require a fixed bound on the
arity of relations.  Furthermore, they  bound the amount of data that
is associated with a single node in a dtgd, in various ways. To state this more precisely, we use the following notions. 
\begin{definition}[bounded context]
	The \emph{context} $\context{\node}(\atoms)$  of a node variable~\node in a set~\atoms of atoms is the set %
	of (data) variables occurring in atoms referring to~\node. The
        context $\context{\node}(\sigma)$ of $\kappa$ in a dtgd
        $\sigma$ is
        $\context{\node}(\rbody{\sigma}\cup\head{\sigma})$. The
        context $\context{\node}(\sigma)$ of $\kappa$ in a degd
        $\sigma$ is
        $\context{\node}(\rbody{\sigma})$. %

        A node variable~\node has \emph{$b$-bounded  context} in $\sigma$ if $|\context{\node}(\sigma)|\le b$. It has \emph{$b$-bounded body context} if $|\context{\node}(\rbody{\sigma})|\le b$. %
      \end{definition}
	  For instance, in the two following constraints,
	  \begin{itemize}
		  \item dtgd
			  \begin{math}
				  \sigma_1 = R(x,y,z), S(y)\at\vark \to T(x)\at\vark
			  \end{math}
			  and
		  \item degd
			  \begin{math}
				  \sigma_2 = S(x)\at\vark, S(y)\at\vark, R(x,y,z)\at\varl \to \vark=\varl,
			  \end{math}
	 \end{itemize}
	 node variable~\vark has context~$\{x,y\}$ and thus 2-bounded context. The body context of~\vark in~$\sigma_1$ is even 1-bounded.
	 Note that, in~$\sigma_2$, node variable~\varl has 3-bounded body context and that, since there is no other node variable, the body context of this constraint is bounded by $3 = \max\{2,3\}$ in general.

      We sometimes simply speak of \emph{bounded context} if $b$ is clear from the, well, context.

	  \subsection{Classes with \texorpdfstring{\pitwo}{Pi_2}-reasoning}
\label{sec:np-reasoning}

In this subsection, we consider two fragments which allow reasoning in
\pitwo in general, and in~\NP,  if there are no comparison atoms. %

 The first fragment requires only
 one restriction (besides the usual arity restriction). The \emph{bounded generation} fragment \Tbg[b] allows all global and data-collecting dtgds but only node-creating dtgds, in which the head variable has $b$-bounded context. We refer to the latter as  node-creating dtgds of Type (G1), cf.\ Table~\ref{table:fragments}. 

\begin{theorem}
\label{thm:ImplicationForBoundedGenerationNP}%

For fixed $\alpha \ge 1$ and $b \ge 1$, problem $\Implication_\alpha(\Tbg[b] \cup \Eall)$ is
\begin{enumerate}
	\item \pitwo-complete in general and
	\item \NP-complete, if restricted to inputs without comparison atoms.
\end{enumerate}
\end{theorem}
\begin{proof}[Proof idea]
	The lower bounds follow by reductions from the containment problem for conjunctive queries (with or without comparisons)~\cite{DBLP:journals/jcss/Meyden97,DBLP:conf/stoc/ChandraM77}. For two queries~\qr and~\qrA of the respective classes, query~\qr is contained in~\qrA if and only if $\{\sigma\} \models \tau$, where $\sigma = \body{\qrA} \to \head{\qrA}$ and $\tau = \body{\qr} \to \head{\qr}$ are considered as global data-collecting dtgds (with or without comparisons).

	\smallskip
	The proofs of the upper bounds use Condition~(3) from Proposition~\ref{prop:chase-correct} and rely on the fact that, in each chase sequence, thanks to the (G1)-restriction only a polynomial number of nodes is generated and thanks to the arity restriction, each can carry only a polynomial number of facts. The \pitwo upper bound can be almost directly inferred from the quantifier structure of Condition~(3). The \NP upper bound follows since, essentially, only one initial database needs to be considered. Below, we provide the details.

We begin by showing that all possible chase sequences have polynomial length.
  Let $\Sigma$ be a set of dependencies in $\Tbg[b]\cup \Eall$ and $\tau$
 be a distribution constraint. We recall that the chase applies a
 node-creating chase step with a dtgd $\sigma$ and a valuation $W$ only if no node with the
 facts from $W(\head{\sigma})$ exists.   However, for each
 $\sigma\in\Sigma$,  the number of variables in $\head{\sigma}$
 {occurring in atoms related to $\kappa$} is at most $b$ and thus the number of different valuations of $\head{\sigma}$ (with the initial values derived from %
{$D_\tau$}) is at most {$|\dom(\Sigma,\tau,\domain)|^b$}, and thus
polynomial. Therefore, the number of chase steps using a $\sigma$ of
Type (G1) is
polynomially bounded.  In particular, the chase generates only a polynomial number of nodes.  
 {Since data-collecting dtgds have only one atom in their head and $\alpha$ is a bound on the arity of atoms, there can only be a polynomial number of chase steps using data-collecting dtgds, for each node. Similarly, there can only be a polynomial number of chase steps using global dtgds.}
  Altogether there can be only a polynomial number of chase steps in
  each chase sequence.
  As mentioned before, the \pitwo upper bound follows from Condition~(3) in Proposition~\ref{prop:chase-correct}. Universal
  quantification is over all databases in
  $\calD(\Sigma,\tau,\domain)$, the chase sequence \chseq is
  existentially quantified and that it fails or there exists an
  appropriate extension can be verified by further existential
  quantification.

  If there are no comparison atoms in $\Sigma$ and $\tau$ it suffices
  to start the chase from \emph{one} canonical database of the form
  $V(\body{\tau})$, for some one-one valuation $V$ that does not map
  any variables of $\body{\tau}$ to constants of $\body{\tau}$. However, the chase
  needs to be defined in a slightly different fashion: if  a
  degd with a head of the form $t=t'$ is applicable via a valuation
  $W$ then in the result  $W(t)$ and $W(t')$ are identified, unless
  they are different constants from $\body{\tau}$. If the latter is
  the case, the chase fails. For this version of the chase,
  Proposition~\ref{prop:chase-correct} holds as well.
 
 The \NP upper bound then follows, since only one initial database needs to
 be used and only one chase sequence of polynomial length needs to be guessed. 
\end{proof}

\medskip\noindent
The other class of dtgds considered in this subsection allows node-creating dtgds with head variables with unbounded context.  The simple argument of the proof of Theorem~\ref{thm:ImplicationForBoundedGenerationNP} therefore does not work anymore. However, it turns out that there are simple (and still generous) restrictions that guarantee a \pitwo (\NP) upper bound for the implication problem.
 To this end, we define the \emph{bounded context fragment}  \Tbd[b]
 of dtgds as follows (cf. Table~\ref{table:fragments}).
 \begin{definition}[bounded context dtgds]
  A node-creating dtgd $\sigma$ is in \Tbd[b], if 
   \begin{itemize}
	   \AdjustDefinitionItemizationIndent
 \item[(G1)] its head variable  has $b$-bounded  context, or
\item[(G2)] all node variables in its body have $b$-bounded  context.
 \end{itemize}
 A data-collecting dtgd $\sigma$ is in \Tbd[b], if
   \begin{itemize}
	   \AdjustDefinitionItemizationIndent
\item[(C1)]  its head variable  has $b$-bounded  body context, or
\item[(C2)] all other node variables have $b$-bounded  context.
\end{itemize}
For instance, all global dtgds and degds are in \Tbd[b].
 \end{definition}

\noindent
The \emph{bounded context fragment}  \Ebd[b] of degds is defined similarly.
\begin{definition}[bounded context degds]
  A degd $\sigma$ with only data variables in its  head is in $\Ebd[b]$. A degd \mbox{$\sigma=\atoms \to \kappa=\mu$} is in  \Ebd[b], if
  \begin{itemize}    
	   \AdjustDefinitionItemizationIndent
  \item[(E1)] $\kappa$ and $\mu$ have $b$-bounded context, or
  \item[(E2)] $\mu$ and all node variables that do not occur in the head have   $b$-bounded context.
  \end{itemize}
\end{definition}
The degds of \Ebd[b] are illustrated in Table~\ref{table:fragments}.
In degds of Type (E2), we call $\kappa$ (but not $\mu$) the \emph{head variable}.

We can now state the second result of this subsection.

\begin{theorem}
	\label{thm:ImplicationForBoundedTransferIsNPComplete}%
	\label{thm:ImplicationForLinearlyUnboundedTransferIsNPComplete}%

		For fixed $\alpha \ge 1$ and $b \ge 1$, problem $\Implication_\alpha(\Tbd[b] \cup \Ebd[b])$ is
		\begin{enumerate}
			\item \pitwo-complete in general and
			\item \NP-complete,  if restricted to inputs without comparison atoms.
		\end{enumerate}
\end{theorem}

\begin{proof}[Proof idea]
  For the upper bounds, we show that any chase sequence for
  $\Tbd[b]\cup \Ebd[b]$ can be \emph{normalised} such that only a
  polynomial number of \emph{witness nodes} are needed to trigger any
  chase steps. Since every node has only a polynomial number of facts,
  this implies that it suffices to consider chase sequences of polynomial length. %
  The remaining arguments are then as for Theorem~\ref{thm:ImplicationForBoundedGenerationNP}.
  The lower bounds follow by the same reduction as
  in
  Theorem~\ref{thm:ImplicationForBoundedGenerationNP}.
\end{proof}

\subsection{Classes with \texorpdfstring{\PSPACE}{PSPACE}-reasoning}
\label{sec:pspace-reasoning}
 
In this subsection, we consider a fragment of distribution constraints that does not guarantee polynomial-length chase sequences but, intuitively, sequences of polynomial \enquote{width}. Consequently, their implication problem turns out as \PSPACE-complete.

The fragment $\Twbd[b]$ is defined as follows (cf.\ Table~\ref{table:fragments}). %
\begin{definition}[weakly bounded distribution tgds]
  Let $b\ge 1$.
A dtgd $\sigma$ is in the class $\Twbd[b]$ of \emph{weakly bounded distribution tgds} if it is in  $\Tbd[b]$ or it obeys the following restriction:
\begin{itemize}
	   \AdjustDefinitionItemizationIndent
	\item[(G3)] $\sigma$ is node-creating and exactly one of its node variables does \emph{not} have $b$-bounded body context.
\end{itemize}
\end{definition}

\begin{theorem}\mbox{ }
	\label{thm:ImplicationForWeaklyBoundedTransferIsPSPACEcomplete}%
	\begin{enumerate}
		\item $\Implication_\alpha(\Twbd[b]\cup \Ebd[b])$ is in \PSPACE, for every $\alpha\ge 1$  and $b\ge 1$. 
		\item $\Implication_\alpha(\Twbd[b])$ is \PSPACE-hard
                  for $\alpha\ge 1$ and $b \ge 0$.   This lower bound
                  even holds without comparison atoms.%
	\end{enumerate}
      \end{theorem}

      \begin{proof}[Proof idea]
        The lower bound~(2) is shown similarly as
        \PSPACE-hardness of the implication problem for inclusion
        dependencies over schemas of \emph{unbounded}
        arity~\cite{DBLP:books/aw/AbiteboulHV95,DBLP:journals/jcss/CasanovaFP84}. In a nutshell, in
        this reduction each node carries \emph{one tuple}, encoded
        with unary relations.

For the upper bound,        unlike for \Tbd, we do not have a
polynomial length bound for chase sequences for \Twbd. In fact, it
might be the case that a chase sequence generates an exponential
number of nodes. However, we can still  use a polynomially bounded set
$Z$ of witness nodes
for the bounded node variables of (G3) dtgds and for all other
constraints. They do not account for the unbounded node variables in
(G3) constraints, but we show that those only need to occur in linear
succession.  The basic idea of the algorithm is to guess $Z$ (and the
facts on nodes from $Z$) and to verify in polynomial space, for each
node in $Z$, that it is produced by a chase sequence. These verifying
computations all assume the same set $Z$. We use a kind of timestamps
to avoid cyclic reasoning.  
		The details of this proof are given in Appendix~\ref{app:PSPACEReasoning}.
      \end{proof}

\subsection{Classes with \texorpdfstring{\EXPTIME}{EXPTIME}-hard reasoning}\label{sec:lowerbounds}

In this subsection, we turn to combinations of constraints that yield an \EXPTIME-hard implication problem. In particular, we complete the proof of Theorem~\ref{thm:ImplicationForDataFullDTGDisEXPTIMEComplete}.  
To this end, we consider the following additional types of constraints:
\begin{itemize}
	   \AdjustDefinitionItemizationIndent
   \item[(G4)]  node-creating dtgds with two unbounded node variables;
\item[(C3)] data-collecting dtgds with two unbounded node variables;
   \item[(E3)] degds with two unbounded node variables; and,
    \item[(E4)] degds with three unbounded node variables.
\end{itemize}

\begin{theorem}\label{prop:lowerbounds}
  $\Implication(\Tdf)$ is  $\EXPTIME$-hard. This statement holds
  already without comparison atoms and with only the following combinations
  of constraint types allowed: 
  \begin{itemize}
  \item[(a)] Node-creating dtgds of Type (G2) and data-collecting dtgds of Type (C3);
  \item[(b)] Node-creating dtgds of Type (G2) and (G4); %
  \item[(c)]  Node-creating dtgds of Type (G2) and degds of Type (E4);
 \item[(d)]  Node-creating dtgds of Types (G2) and (G3), and degds of Type (E3).
 \end{itemize}
 In all cases,  schemas with (at most) binary relations suffice.
\end{theorem}
The four \EXPTIME-hard fragments are illustrated in
Table~\ref{table:fragments}. The reductions use an alternating Turing machine with \emph{linearly} bounded space. 

\subsection{Parallel-correctness revisited}

We lift parallel-correctness to the setting of distribution constraints. In particular, we say that a query $Q$ is \emph{parallel-correct w.r.t.\ a set of distribution constraints} $\Sigma$ if $Q$ is parallel-correct on every database that satisfies~$\Sigma$.%

As parallel-correctness of a conjunctive query can be expressed as a dtgd, the results of the present section lead to the following:

\begin{corollary}\label{coro:pc}
	For a CQ $Q$ and a set of distribution constraints $\Sigma$,
        the complexity of deciding parallel-correctness of $Q$ w.r.t.\
        $\Sigma$ is in \EXPTIME. Furthermore, it is in
        \pitwo (or  \NP, without comparsion atoms) and \PSPACE if $\Sigma\subseteq \Tbd[b]\cup \Ebd[b]$ and
$\Sigma\subseteq \Twbd[b]\cup \Ebd[b]$, respectively, for a fixed $b$ and a fixed bound $\alpha$ on the maximal arity of relation symbols.
\end{corollary}

\section{Conclusion}
\label{sec:Conclusion}

In this work, we introduced a novel declarative framework based on classical tgds and egds with comparison atoms to specify and reason about classes of data distributions. We illustrated our framework by various examples and performed an initial study of the complexity of the implication problem. As an application, we derived bounds (in Corollary~\ref{coro:pc}) for the complexity of parallel-correctness
of conjunctive queries.

Of course, there are many immediate general directions for extending the line of work started in this paper. For instance, one could 
study the implication problem for more expressive distribution constraints than data-full ones. There is a plethora of work on fragments of dependencies for improving the complexity of decision problems~(e.g.,\cite{DBLP:conf/ijcai/BagetMRT11,DBLP:conf/pods/Benedikt18,DBLP:journals/jair/CaliGK13,Onet12}). 
It could be investigated if any of these or others lead to a decidable implication problem. Another direction for future work is to study parallel-correctness w.r.t.\ distribution constraints for more expressive query languages than conjunctive queries. Some possibilities are unions of conjunctive queries~\cite{DBLP:journals/jacm/AmelootGKNS17}, conjunctive queries with negation~\cite{DBLP:conf/icdt/GeckKNS16} or Datalog~\cite{DBLP:conf/icdt/KetsmanAK18}. 

Example~\ref{example:lineitem} and Section~\ref{sec:hierar} illustrate how co-partitioning schemes can be translated into distribution constraints. It would be interesting to investigate the converse direction. That is, by design, distribution constraints specify in a declaratively way which properties a horizontal partitioning should satisfy. They do not provide a direct operational way to compute an actual partitioning. A natural question is to find an optimal partitioning satisfying a given set of distribution constraints.

Section~\ref{sec:constraints} mentions a translation of distribution constraints to classical tgds and egds by increasing the arity of relations by one to take the node variables into account. It would be interesting to see whether the resulting fragment of dependencies is worthwhile to study it on its own in the classical setting.

The main technical challenge left open from this work is whether the \EXPTIME-hardness result in Theorem~\ref{prop:lowerbounds}(c) can be extended to rules of Type (E4) that contain two rather than three unbounded node variables. 

\begin{table}
  \centering
  \newcommand{\thickLine}{1.3pt}%
  \newcommand{\thickVerticalLine}{!{\vrule width\thickLine}}%
  \newcommand{\thinHorizontalLine}{}%
  \newcommand{\xmark}{\checkmark}%
  \begin{tabular}{cl!{\vrule width\thickLine}c|c!{\vrule width\thickLine}c!{\vrule width\thickLine}c|c|c|c}
	  &&\multicolumn{2}{c!{\vrule width\thickLine}}{\pitwo (\NP)}& \PSPACE &\multicolumn{4}{c}{\EXPTIME}\\
	  \midrule[\thickLine]
  	\textsf{(G1)} & $\unbounded{\nodeA_1} \dotbox \unbounded{\nodeA_r}  \to \bounded{\node}$&\xmark&\xmark&\xmark&&&&\\\thinHorizontalLine
  \textsf{(G2)} & $\bounded{\nodeA_1} \dotbox \bounded{\nodeA_r}  \to \unbounded{\node}$&&\xmark&\xmark&\xmark&&\xmark&\xmark\\\thinHorizontalLine
  \textsf{(G3)} & $\bounded{\nodeA_1} \dotbox \bounded{\nodeA_r} \unbounded{\mu}  \to \unbounded{\node}$&&&\xmark&&&&\xmark\\\thinHorizontalLine
 \textsf{(G4)} & $\unbounded{\varl} \unbounded{\varm}  \to \unbounded{\node}$ &&&&&\xmark&&\\
 \midrule[\thickLine]
 \multicolumn{2}{l!{\vrule width\thickLine}}{\textsf{Unrestricted data-collecting dtgds}}&\xmark&&&&&&\\\thinHorizontalLine
  \textsf{(C1)} & $\bounded{\node} \unbounded{\nodeA_1} \dotbox \unbounded{\nodeA_r}  \to \unbounded{\node}$&&\xmark&\xmark&&&&\\\thinHorizontalLine
  \textsf{(C2)} & $\unbounded{\node} \bounded{\nodeA_1} \dotbox \bounded{\nodeA_r}  \to \unbounded{\node}$&&\xmark&\xmark&&&&\\\thinHorizontalLine
  \textsf{(C3)} &  $\unbounded{\node} \unbounded{\nodeA}  \to \unbounded{\node}$ &&&&\xmark&&&\\
 \midrule[\thickLine]
 \multicolumn{2}{l!{\vrule width\thickLine}}{\textsf{Unrestricted  degds}}&\xmark&&&&&&\\\thinHorizontalLine
  \textsf{(E1)} & $\bounded{\kappa} \bounded{\mu} \unbounded{\nodeA_1} \dotbox \unbounded{\nodeA_r}  \to \kappa=\mu$&&\xmark&\xmark&&&&\\\thinHorizontalLine
  \textsf{(E2)} & $\unbounded{\kappa} \bounded{\mu} \bounded{\nodeA_1} \dotbox \bounded{\nodeA_r}  \to \kappa=\mu$ &&\xmark&\xmark&&&&\\\thinHorizontalLine
  \textsf{(E3)} & $\unbounded{\node} \unbounded{\nodeA}  \to \node = \nodeA$&&&&&&&\xmark\\\thinHorizontalLine
  \textsf{(E4)} & $\unbounded{\node} \unbounded{\nodeA} \unbounded{\nodeB}  \to \node = \nodeA$&&&&&&\xmark&\\
	  \midrule[\thickLine]
 & Theorem & \ref{thm:ImplicationForBoundedGenerationNP} & \ref{thm:ImplicationForBoundedTransferIsNPComplete} & \ref{thm:ImplicationForWeaklyBoundedTransferIsPSPACEcomplete} & \ref{prop:lowerbounds}(a)  & \ref{prop:lowerbounds}(b)  & \ref{prop:lowerbounds}(c)  & \ref{prop:lowerbounds}(d)\\ 
    \bottomrule[\thickLine]
\end{tabular}
\caption{Illustration of restricted classes of dtgds and degds. Node variables that have to be bounded are shaded, others may be unbounded. The columns indicate the complexity of (some) combinations of fragments.}
\label{table:fragments}
\end{table}

\newpage
\bibliographystyle{abbrv}
\bibliography{references.bib}
\clearpage
\section*{Appendix}
\appendix

\renewcommand{\thesection}{A}
\section{Examples for Section~\ref{sec:constraints}}

\subsection{More applications of distribution constraints}
  \label{sec:more:appl:dtgd}

  \subsubsection{Strong Parallel-Correctness}
  A query is \emph{strongly parallel-correct}\footnote{The qualification
    \emph{strong} stems from the fact that this condition is sufficient
    but not necessary for parallel-correctness
    \cite{DBLP:journals/jacm/AmelootGKNS17}.} for a distributed instance
  if, for each valuation $V$ that derives some result tuple, all facts
  in $V(\rbody{Q})$ meet at some node
  \cite{DBLP:journals/jacm/AmelootGKNS17}. Strong parallel-correctness can be captured by data-full dtgds as we exemplify next.

   \begin{example}
  Strong parallel-correctness of the query 
  $H(n,s) \gets \Emp(n,t),\Sal(t,s)$
  can be expressed by the dtgd
  \begin{math}
    \ShrinkSuggestion
    \Emp(x,y), \Sal(y,z) \to  \At{\Emp(x,y)}{\kappa}, \At{\Sal(y,z)}{\kappa}. 
  \end{math}
  Here, only the node variable $\kappa$ is quantified in the head. So, the dtgd is data-full.
  \qed        \end{example}

  We note that strong parallel-correctness of a CQ $Q$ is essentially the basic property that is guaranteed by a hypercube distribution based on $Q$ if nothing is known about the actual hash functions. We discuss hypercube distributions next.

  \subsubsection{More on hierarchical partitioning schemes}
  \label{app:HierarchicalPartitioning}
The AdWords example for F1~\cite{DBLP:journals/pvldb/ShuteVSHWROLMECRSA13}, on relations for customers, advertising campaigns and adword groups, can be modeled as follows:
\begin{eqnarray*}
    \Customer(\custId,\bx)\at\vark, \Campaign(\custId,\campId,\by) \to \Campaign(\custId,\campId,\by)\at\vark, \\
    \Campaign(\custId,\campId,\by)\at\vark, \AdGroup(\custId,\campId,\bz) \to \AdGroup(\custId,\campId,\bz)\at\vark.
\end{eqnarray*}

However, the dtgd framework allows to specify more advanced co-hashing strategies by using multiple relations in the body of a dtgd. For instance, the hierarchical distribution above could be enforced when some information of a supplier with the same nation key as the customer is available at a node, as in the dtgd below:
\begin{eqnarray*}
    \Customer(\custId,\nation,\bx)\at\vark, \Campaign(\custId,\campId,\by), \Supplier(\supId,\nation,\bz)\\  \hfill{} \to \Campaign(\custId,\campId,\by)\at\vark.
\end{eqnarray*}
Furthermore, dtgds do not require the head to refer to one of the relations referred to in the body, thus allowing to model the co-hashing of mere summaries (of relations or joins over relations) like in relation $\Campaign^*$ in
\begin{displaymath}
	\Customer(\custId,\nation,\bx)\at\vark, \Campaign(\custId,\campId,\by) \to \Campaign^*(\custId,\campId)\at\vark.
\end{displaymath}

  \subsubsection{Hypercube Distributions}
\label{app:hc}

  Consider the query $Q=H(u,x,y,w) \gets R(u,x), S(x,y), T(y,w)$. The hypercube algorithm \cite{DBLP:journals/tkde/AfratiU11,DBLP:journals/jacm/BeameKS17} evaluates this query in the MPC framework over a two-dimensional network $\nw = \{1,\dots,p_1\} \times \{1,\dots,p_2\}$ using some hashing functions $h_1:\dom\to\{1,\dots,p_1\}$ and $h_2:\dom\to\{1,\dots,p_2\}$, respectively, such that in distribution $\distH_Q$
  \begin{align*}
    R(a,b) & \text{ is mapped to all nodes } (h_1(b),\star); \\
    S(b,c) & \text{ is mapped to all nodes } (h_1(b),h_2(c)); \text{ and,} \\
    T(c,d) & \text{ is mapped to all nodes } (\star,h_2(c)). 
  \end{align*}
  The hypercube distribution~$\distH_Q$ is designed such that the initial query~$Q$ is strongly parallel-correct under it. Thus, for every valuation~$V$ for~$Q$, there is some node where the facts required by~$V$ meet. Therefore, $\distH_Q$ satisfies the global-local constraint
  \begin{displaymath}
    \sigma_Q = R(u,x), S(x,y), T(y,w) \to \At{R(u,x)}{\vark}, \At{S(x,y)}{\vark}, \At{T(y,w)}{\vark}.
  \end{displaymath}

  However, dependency~$\sigma_Q$ covers only a small aspect of a
  hypercube distribution for~$Q$. This already becomes clear for a query
  like $Q'=H'(u,x,y) \gets R(u,x), S(x,y)$, for which parallel-correctness
  transfers\footnote{We say that parallel-correctness
    \emph{transfers} from query $Q$ 
  to query $Q'$ if $Q'$ is parallel-correct under every distribution
  policy for which $Q$ is
  parallel-correct~\cite{DBLP:journals/jacm/AmelootGKNS17}.}  from~$Q$.
  Distribution~$\distH_Q$ satisfies global-local constraint
  \begin{displaymath}
    \sigma_{Q'} = R(u,x), S(x,y) \to \At{R(u,x)}{\vark}, \At{S(x,y)}{\vark},
  \end{displaymath}
  which is not implied by~$\sigma_Q$ because instances with missing $T$-facts are not guaranteed to be parallel-correct for~$Q'$. Furthermore, there are queries that are parallel-correct under~$\distH_Q$ although parallel-correctness does not transfer from~$Q$ to them. One such example is $Q''=H''(u_1,u_2,x) \gets R(u_1,x), R(u_2,x)$, whose corresponding dtgd
  \begin{displaymath}
    \sigma_{Q''} = R(u_1,x), R(u_2,x) \to \At{R(u_1,x)}{\vark}, \At{R(u_2,x)}{\vark}
  \end{displaymath}
  is again not implied by~$\sigma_Q$.

  As an example, we describe how 2-dimensional hypercube distributions can be modeled by distribution constraints. Technically, due to the absence of functions in distribution constraints, we neglect the meeting of facts in a hypercube distribution that is solely caused by collisions under the hash functions. This can be viewed as reasoning about all 2-dimensional hypercube distributions or an \enquote{abstract} hypercube distribution over network $\nw=\dom\times\dom$.

  The modeling relies on two auxiliary relations. Unary relation~\Dom is intended to contain all data values of the global database. A distributed fact $H(a,b)\at\nodek$ is intended to represent the mapping of $(a,b)$ to node~\nodek by the pair $(h_1,h_2)$ of hashing functions.

  For each relation symbol~$R$ of arity~$r$, a dtgd
  \begin{math}
    \ShrinkSuggestion
    R(x_1,\dots,x_r) \to \Dom(x_j)
  \end{math}
  is added for every $j \in \{1,\dots,r\}$ capturing the semantics of the $\Dom$-predicate. Then, a single generating dtgd
  \begin{math}
    \ShrinkSuggestion
    \Dom(x), \Dom(y) \to \At{H(x,y)}{\vark}
  \end{math}
  is added such, for each pair of data values, there is at least one node responsible for them.

  Finally, each of the hashing rules is described by a single collecting dtgd,
  \begin{align*}
    R(u,x), \Dom(z), \At{H(x,z)}{\vark} & \to \At{R(u,x)}{\vark}, \\
    S(x,y), \At{H(x,y)}{\vark} & \to \At{S(x,y)}{\vark}, \\
    T(y,w), \Dom(z), \At{H(z,y)}{\vark} & \to \At{T(y,w)}{\vark}.
  \end{align*}

  In particular, the resulting set of distribution constraints has bounded context as defined in Section~\ref{sec:np-reasoning}.

  \subsubsection{Naive query answering}
  \label{sec:subsec:qa}
  Next, we introduce query answering in the context of distribution constraints adopting certain answers as the underlying semantics. We stress that this is only one possible way of many to define certain answers.
  Given an instance $I$, a set of dependencies $\Sigma$, and a query $Q$, the certain answers $\text{certain}(Q,I,\Sigma)$ are defined as those facts that are selected by the naive evaluation of $Q$ over every distributed instance that is consistent with $I$ and $\Sigma$. Formally, $$\text{certain}(Q,I,\Sigma)= \bigcap_D \{\Qnaive(D) \mid  \myglobal(D) = I \text{ and } D \models \Sigma\}.$$

  \begin{example}
  Let
  $I=\{\Emp(a,t),\Emp(a',t'),\Sal(t,s_1),\Sal(t,s_2),\allowbreak\Sal(t',s')\}$,
  and let $\Sigma$ consist of the non-skipping constraints
  $\calU(\schema)$ and the single dtgd 
  \begin{math}
    \ShrinkSuggestion
    \At{\Emp(x,y)}{\kappa}\to \At{\Sal(y,z)}{\kappa},
  \end{math}
  and let $Q=H(x,z) \gets \Emp(x,y),\Sal(y,z)$. 
  Then
  \begin{math}
    \ShrinkSuggestion
    \text{certain}(Q,I,\Sigma) = \{Q(a',t')\}.
  \end{math}
  \qed
  \end{example}

  \subsubsection{Multi-round communication}
  \label{sec:subsec:multi-round}

  Parallel-correctness---as defined in the previous sections---is set within the single-round communication model where each node naively evaluates the same query over its local database. We exemplify how distribution constraints can be adapted to incorporate (say, query evaluation in) the multi-round communication model where data can be reshuffled in between rounds.

  For this, we assume there is a constant number of rounds (or an upper bound on that number). The basic idea is that for every relation symbol $R$ constraints can use atoms of the form $R^{(i)}$ referring to the contents of relation $R$ in the $i$-th communication round. For instance, the rule $\At{R^{(i)}(x,y)}{\kappa}\to \At{R^{(i+1)}(x,y)}{\kappa}$ expresses that every fact $R(a,b)$
  occurring on a node at round $i$ is also at that node for round $i+1$.
  \begin{example}
  Consider the query 
  \begin{math}
    \ShrinkSuggestion
    Q = H(x,z) \gets R(x,y),S(y,z),T(z,x).
  \end{math}
  The following constraints are consistent with the two stage evaluation of $Q$ that first evaluates $O(x,z) \gets R(x,y),S(y,z)$ in the first round and $O(x,z),T(z,x)$
  in the second round:
  \begin{eqnarray*}
  R(x,y),S(y,z) \to \At{R^{(1)}(x,y)}{\kappa},\At{S^{(1)}(y,z)}{\kappa}\\
  \At{O^{(1)}(x,z)}{\kappa}, T(z,x) \to \At{O^{(2)}(x,z)}{\kappa'},\At{T^{(2)}(z,x)}{\kappa'}
  \end{eqnarray*}
  Here, $O$ is viewed as an intensional relation that is computed during the first computation round. \qed
  \end{example}
  With this translation, which relies on an \emph{a priori} known number of rounds, our upper bounds hold unchanged. It remains, however, unclear how multiple rounds can be modelled without an \emph{a priori} bound on the number of rounds.

\renewcommand{\thesection}{B}
\section{Missing proofs for Section~\ref{sec:implication}}

\begin{proof}[Proof sketch of Proposition~\ref{prop:chase-terminates}]
Termination follows from two simple observations.
\begin{itemize}
\item The constraints from $\Sigma$ never introduce any new data
  values and thus, for each rule $\sigma$, the number of valuations of
  data variables from  $\body{\sigma}$ is finite (in fact at most
  exponential in $|D\cup\data(\Sigma)|$, where $\data(\Sigma)$
  denotes  the set of data values in $\Sigma$).
\item Each constraint  $\sigma\in \Sigma$ fires at most once, for each
  valuation of
  data variables from  $\body{\sigma}$.
\end{itemize}
\end{proof}

\begin{proof}[Proof sketch of Proposition~\ref{prop:chase-correct}]
  We first show the equivalence of (1) and (2) by two contrapositions. Clearly, if (2)
  fails due to a chase sequence \chseq then  $\Chase(\chseq)$  witnesses
  that (1) fails. We thus show in the following that failure of (1) implies
  failure of (2).

  To this end, let $D'$ be a distributed database with $D'\models
  \Sigma$ and $D'\not\models\tau$ and let $V$ be a valuation of
  $\var(\body{\tau})$ that witnesses  $D'\not\models\tau$.
   Let  $D$ be  $V(\rbody{\tau})$ and let \chseq be a maximal chase sequence
   for $D$ with $\Sigma$. By induction on the number of
   steps, it is easy to show that, for each prefix  $\chseq'$ of
   $\chseq$, there is a homomorphism from $\Chase(\chseq')$ to $D'$
   that is the identity on $\dom$. We conclude that $\chseq$ is
   successful, since if it failed due to some degd, $D'$ would also violate
   that degd. Furthermore, $V$ satisfies $\body{\tau}$ on
   $\Chase(\chseq)$. If $V$ could be extended to a satisfying
   valuation $V'$ of $\head{\tau}$ on $\Chase(\chseq)$, this would
   also be possible on $D'$. Thus, $\Chase(\chseq)\not\models\tau$.

   By construction of $\calD(\Sigma,\tau,\domain)$, there is a
   database $D''\in \calD(\Sigma,\tau,\domain)$ that is isomorphic to
   $D$, even with respect to the linear order, and therefore $D''$
   witnesses the failure of (2).

   The equivalence of  (2) and (3) follows immediately from
   \cite[Theorem 3.3]{DBLP:journals/tcs/FaginKMP05} and \cite[Prop.\
   2.6]{DBLP:journals/tcs/FaginKMP05}. If some $D$ has a failing chase
   sequence then it has no successful chase sequence at all. And the results
   of all successful chase sequences are homomorphically equivalent. Therefore it
   suffices to consider, for each $D\in \calD(\Sigma,\tau,\domain)$,
   only one chase sequence.
\end{proof}

\renewcommand{\thesection}{C}
\section{Missing proofs for Section~\ref{sec:complexity}}

In lower bound proofs, we depart from the convention that dtgds have
only one head atom. We thus allow ourselves to use data-collecting dtgds
with more than one atom in their head. This is only a convenience as
explained earlier.

\subsection{Proof for classes with bounded context}\label{app:BoundedContext}
\begin{proof}[Proof details for Theorem~\ref {thm:ImplicationForBoundedTransferIsNPComplete}]
  We show the upper bounds in more detail.
  
Let, in the following,   $\alpha\ge 1$ and $b\ge 1$, and $\Sigma$ be a
set of dtgds from $\Tbd[b]\cup \Ebd[b]$ and $\tau$ any distribution
constraint. Let, furthermore, $D \in \calD(\Sigma,\tau,\domain)$ be of the
form $V(\body{\tau})$ and let
\chseq be a chase sequence for $D$. We show that  there is a chase sequence %
    of length at most $p(\norm{\Sigma},|\tau|)$ which behaves like
    \chseq in the following sense: either they both fail or they both allow an extension of
    $V$ that satisfies $\head{\tau}$ or they both do not allow such an
    extension.
	Furthermore, the degree of $p$ only depends on $\alpha$ and $b$.

   To this end, let %
   $D$ and $\chseq$ be as above. In the first proof step, we define a new, \enquote{normalised} chase sequence $\chseq'$ of the same length in an inductive fashion. In the second step, we extract a polynomial size chase sequence $\chseq''$ from $\chseq'$.

  The idea of the normalisation step is to bound the number of
  {witness nodes} which are used in the chase sequence to trigger
  constraints. A node $k$ is a \emph{witness node} for a chase step
  with dtgd $\sigma$ and valuation $W$ if $W(\lambda)=k$ for some node variable $\lambda$ of $\sigma$ that is \emph{not} the head variable.

In a nutshell, if, for some constraint $\sigma$ of Type (G2), (C2) or (E2), a bounded non-head node variable $\lambda$ already had a witness node with the same valuation as $\lambda$ before, then the earliest such witness node is used again. Orthogonally, if for a constraint $\sigma$ of Type (C1) or (E1) the same valuation for the head variable (or for $\kappa$ and $\mu$ for Type (E1)) has occurred before, then the witness nodes of the earliest such occurrence are used again for the current chase step. 

We now inductively  describe the construction of $\chseq'$ in more detail. Let
$\sigma_i$ and $W_i$ denote the dtgd and valuation of the $i$-th chase
step in \chseq. For $\chseq'$, we use the same distribution constraints but possibly
different valuations $W'_i$. 
We let $W'_1=W_1$.

For $i> 1$, we define $W'_i$ as follows, depending on the type of $\sigma_i$: %

 \begin{itemize}
 \item If $\sigma_i$ is a dtgd of Type (G1), a global dtgd or a degd with data variables in its head, then $W'_i=W_i$.
 \item If  $\sigma_i$ is a dtgd of Type (G2) or (C2) or a degd of Type (E2), $W'_i$ is defined as follows:
   \begin{itemize}
   \item If for some non-head node variable\footnote{This includes $\mu$, in the case of degds.} $\lambda$, there is some $j<i$ with $\sigma_j=\sigma_i$, such that  $W_j(x)=W_i(x)$, for all $x\in \context{\lambda}(\sigma_i)$, then $W'_i(\lambda)=W_j(\lambda)$, for the smallest such $j$. 
   \item For all other (node or data) variables $x$, let $W'_i(x)=W_i(x)$. 
\end{itemize}
\item  If  $\sigma_i$ is a dtgd of Type (C1) with head variable $\kappa$, valuation~$W'_i$ is defined as follows:
  \begin{itemize}
  \item If there is some $j<i$ with $\sigma_j=\sigma_i$, such that  $W_j(x)=W_i(x)$, for all $x\in \context{\node}(\sigma_i)$, then $W'_i(\lambda)=W_j(\lambda)$ for all node variables $\lambda\not=\kappa$ and  $W'_i(x)=W_j(x)$, for all $x\in \context{\lambda}(\atoms)$, for the smallest such $j$.  
  \item For all other (node or data) variables $x$, let $W'_i(x)=W_i(x)$.
  \end{itemize}
\item  If  $\sigma_i$ is a degd \mbox{$\atoms,\comparisons \to \kappa=\mu$} of Type (E1), but not of Type (E2), $W'_i$ is defined as follows:
  \begin{itemize}
  \item If there is some $j<i$ with $\sigma_j=\sigma_i$, such that  $W_j(x)=W_i(x)$, for all $x\in \context{\node}(\sigma_j)\cup\context{\mu}(\sigma_i)$, then $W'_i(\lambda)=W_j(\lambda)$ for all node variables $\lambda\not\in\{\kappa,\mu\}$ and  $W'_i(x)=W_j(x)$, for all $x\in \context{\lambda}(\atoms)$, for the smallest such $j$.  
  \item For all other (node or data) variables $x$, let $W'_i(x)=W_i(x)$.
  \end{itemize}

 \end{itemize}
 We emphasise that, in all cases, $W'_i$ is well-defined for every data variable~$x$, even if some variables occur in the contexts of multiple node variables.
  
\medskip
Since the normalisation never changes the valuation of variables that
occur in the head of $\sigma_i$, it is not hard to show by induction
on $i$, that $\chseq'$ behaves like $\chseq$ in the above sense. 

 The new sequence $\chseq'$ needs not be sufficiently small, though. However, in the second proof step, we show that we can extract a subsequence $\chseq''$ from $\chseq'$  that still behaves like $\chseq$ and has polynomial size.

 To this end, let the set $Z$ consist of all nodes that occur as witness nodes (as defined above) in some chase step of $\chseq'$ and of a minimal set of nodes that certify the head of $\tau$. We show next that $|Z|=\bigO(\norm{\Sigma}^2|\tau|^{2b+\alpha})$.

In this proof, we always bound the number of data values in $\chseq'$ by $|\tau|$ and the number of variables per constraint, as well as the number of constraints, by $\norm{\Sigma}$. 

\begin{itemize}
\item Each dgtd of Type (G1) can fire at most $|\tau|^b$ times and therefore it has at most  $\norm{\Sigma}|\tau|^b$  witness nodes. Each global dgtd fires at most $|\tau|^\alpha$ times  and  has at most  $\norm{\Sigma}|\tau|^\alpha$ witness nodes.
\item Each degd with data variables in its head can fire at most $|\tau|$ times and   therefore has at most  $\norm{\Sigma}|\tau|$  witness nodes.
\item For each non-head node variable in a constraint of Type (G2), (C2) or (E2), there are at most $|\tau|^b$ valuations of their data variables, and therefore each such constraint needs at most  $\norm{\Sigma}|\tau|^b$  witness nodes.
\item For the data variables of a head variable of a dtgd of Type (C1) there are at most $|\tau|^{b+\alpha}$ valuations, and therefore each such dtgd needs at most $\norm{\Sigma}|\tau|^{b+\alpha}$  witness nodes.
\item For a degd of Type (E1) with head $\kappa=\mu$ there are at most $|\tau|^{2b}$ valuations of the data variables of $\kappa$ and $\mu$, and therefore  each such degd needs at most $\norm{\Sigma}|\tau|^{2b}$  witness nodes. 
\end{itemize}

Since each constraint of $\Sigma$ needs at most $\norm{\Sigma}|\tau|^{2b+\alpha}$ witness nodes, the overall number of witness nodes is at most $\norm{\Sigma}^2|\tau|^{2b+\alpha}$. Since only $|\tau|$ nodes are needed to certify (the head of) $\tau$, we have established the stated bound on $|Z|$.

 Let now $\chseq''$ be the subsequence of $\chseq'$ that contains all
 chase steps where the head variable of the dtgd is mapped to a node
 from $Z$. It is easy to see that $\chseq''$ behaves like $\chseq$ in
 the above sense.
 Since, for each node, the number of possible facts is bounded by
 $\norm{\Sigma}|\tau|^\alpha$, we can conclude that $\chseq''$ has
 polynomial length.

 The upper bounds now follow as in the proof of Theorem~\ref{thm:ImplicationForBoundedGenerationNP}.
\end{proof} 

\subsection{Proof details for classes with \texorpdfstring{\PSPACE}{PSPACE}-reasoning}\label{app:PSPACEReasoning}

      \begin{proof}[Further proof details for Theorem~\ref{thm:ImplicationForWeaklyBoundedTransferIsPSPACEcomplete}]
      
\begin{new}
	We first show
        Theorem~\ref{thm:ImplicationForWeaklyBoundedTransferIsPSPACEcomplete}.1,
        which states that, for every $\alpha \ge $ and $b \ge 1$,
        problem $\Implication(\Twbd \cup \Ebd)$ is in~\PSPACE.
Thanks to Proposition~\ref{prop:chase-correct} it suffices to
construct, for each  database $D=V(\rbody{\tau})$ in $\calD(\Sigma,\tau,\domain)$,
a chase sequence $\chseq $ for $D$ with $\Sigma$ and to check whether
it fails or whether there is an extension $V'$ of $V$ that satisfies $\head{\tau}$ on $\Chase(\chseq)$. 
It suffices to show how this can be done for one such $D$, since a
\PSPACE algorithm can then cycle though all instances from
$\calD(\Sigma,\tau,\domain)$. 

To this end, let in the following $\alpha$ and $b$ be fixed. Let
$\Sigma$ and $\tau$ be from $\Twbd \cup \Ebd$, let $D \in
\calD(\Sigma,\tau,\domain)$ and let $E$ be a set of distributed facts
with data values from $D$. We show the following claim.

\begin{claim}\label{claim:pspace}
  It can be decided in polynomial space (in the size of $D$, $\tau$
  and $E$) whether there is a chase sequence $\chseq$ with $\Sigma$,
  starting from $D$ whose result contains $E$. The polynomial only
  depends on $\alpha$ and $b$.
\end{claim}
We first argue how the upper bound of the theorem follows from this
claim. The algorithm cycles through all distributed instances $D=V(\rbody{\tau})$ from
$\calD(\Sigma,\tau,\domain)$. For each such $D$, it guesses a set $E$, tests in
polynomial space that there is a chase sequence \chseq as in the claim
and accepts if there is an degd in $\Sigma$ that fails over $E$ or if
there is an extension $V'$ of $V$ such that $V'(\head{\tau})\subseteq
E$. The correctness is evident given Proposition~\ref{prop:chase-correct}.

It thus suffices to prove Claim~\ref{claim:pspace} to establish the
upper bound of the theorem. 

A \emph{timed witness set} is a pair $(F,t)$ with a set $F$ of distributed facts and a  \emph{timing function} $t:F\to\Nat$. The intuition behind $t$ is basically to map each fact to the number of the chase step in which it is produced. It thus induces a partial order on $F$. If $t$ is clear from the context, we will usually represent $(F,t)$ just by $F$.
For a node $k$, we write $F[k]$ for the set of distributed facts of the form $f\at{k}$ and $F-k$ for $F-F[k]$. Furthermore, for a natural number $p$, let $F^{< p}$ denote the set of all distributed facts $f\in F$ with $t(f)< p$.

We next consider sequences of  extended chase steps that can use previously produced facts, on one hand, and facts from $F$, on the other hand, but only in a time-respecting fashion. For that purpose these sequences come with a timing function, as well.

 More precisely,  a  dtgd $\sigma$ is  \emph{applicable} to a
 distributed instance $D$ relative to fact set $F$ and time $p$  if it is applicable to $F^{< p}$
 or it is of Type (G3) and applicable to $D\cup F^{< p}$.
 
  A \emph{partial linear chase sequence relative to $F$} 
  is a pair $(\chseq,s)$ with  a  strictly increasing
  \emph{timing function} $s:\{1,\ldots,n\}\to\Nat$ and a sequence
  $\chseq=D_0,D_1,\ldots$ of distributed instances with $D_0=D$, such
  that each $D_{i}$ results from $D_{i-1}$ by an extended chase step
  with a dtgd $\sigma_i$ that is applicable to $D_i$ relative to  $F$
  and time $s(i)$.    A further requirement  is that all produced facts $f$ that are in $ F$ have the same value $t(f)$.

A timed witness set $(F,t)$ is \emph{consistent} with $D$ and
$\Sigma$, if for every node $k$ occurring in $F$ there is a partial
linear chase sequence for $D$ relative to $F-k$ that produces all facts from $F[k]$.

\end{new}

It remains to show the following.
\begin{enumerate}[(i)]
	\item For each chase sequence \chseq for $D$ and each set $E\subseteq\chase{\Sigma}{(\chseq)}$ there is a timed witness
  set $(F,t)$ with  $F$ of polynomial size that is consistent with $D$
  and $\Sigma$ and satisfies $E\subseteq F$. 
\item For each timed witness set $(F,t)$ that is
  consistent with $D$ and $\Sigma$, there exists a chase sequence \chseq starting
  from $D$ with $F\subseteq \chase{\Sigma}(\chseq)$. 
\item There is a polynomial space algorithm that tests, whether for
  $D$ and $E$ there is a timed witness
  set $(F,t)$ with  $F$ of polynomial size that is consistent with $D$
  and $\Sigma$ and satisfies $E\subseteq F$. 
\end{enumerate}

\noindent
Towards (i), let \chseq be a  chase sequence. Similarly as in the
proof of Theorem~\ref {thm:ImplicationForBoundedTransferIsNPComplete}
we can assume that \chseq is normalised according to the rules given in the proof of Theorem~\ref{thm:ImplicationForBoundedGenerationNP}  with the following extension:
\begin{itemize}
\item  If $\sigma_i$ is a dtgd of Type (G3), $W'_i$ is defined as follows:
   \begin{itemize}
   \item If for some non-head node variable $\lambda$, there is some $j<i$ with $\sigma_j=\sigma_i$, such that  $W_j(x)=W_i(x)$, for all $x\in \context{\lambda}(\sigma_i)$, then $W'_i(\lambda)=W_j(\lambda)$, for the smallest such $j$. 
   \item For all other (node or data) variables $x$, let $W'_i(x)=W_i(x)$. 
\end{itemize}
\end{itemize}

The timed witness set $(F,t)$ is constructed as follows. A witness node is a node $k$ that occurs in $\chseq$ in the application of a (G3) rule as $V(\lambda)$, for some bounded variable $\lambda$ or is a witness node for some other step, as defined in the proof of Theorem~\ref{thm:ImplicationForBoundedGenerationNP}. 
Let $Z$ consist of all these witness nodes. 
 Again it holds
$|Z|=\bigO(\norm{\Sigma}^2|\tau|^{2b+\alpha})$ since the same bounds can be shown for non-(G3) constraints, and for each dtgd of Type (G3) there are at most $\norm{\Sigma}|\tau|^b$ witness nodes for the bounded node variables. 

Let $F$ consist of the set of all distributed facts $f$ of nodes in
$Z$ and all facts from $E$. Again, thanks to the arity bound for $\Sigma$, the number of facts per node is polynomially bounded. 
For each fact $f$ in $F$, we let $t(f)$ be the number of the chase step in which $f$ is produced in \chseq.  

It  remains to show that $F$ is consistent with $D$ and $\Sigma$. To
this end, let $k$ be some node occurring in $F$.

For the construction of a partial linear chase sequence for $k$, we
use the concept of (G3)-predecessors. A node $\ell$ is an
\emph{immediate (G3)-predecessor} of a node $\ell'$, for a chase
sequence, if $\ell'$ is generated by a chase step with a
(G3)-dtgd in which the unbounded body node variable is mapped to
$\ell$. The set of (G3)-predecessors of $k$ is obtained by the closure
of $\{k\}$ under immediate (G3)-predecessors. We note that since each
node is generated only once, it can have at most one immediate
(G3)-predecessor, and the set of predecessors induces a linear chain
of nodes. 

Now we are able to define a partial linear chase sequence for $k$. It
consists of all chase steps that generate or contribute facts to $k$
and its (G3)-predecessors. However, if $\ell$ is the immediate
(G3)-predecessor of $\ell'$ only those chase steps producing facts for
$\ell$ are kept which occur in \chseq \emph{before} the generation of
$\ell'$. The timing function $s$ maps each chase step $c$ of the
sequence to the number of this step in \chseq. It is not hard to see
that this construction yields  a partial linear chase sequence for
$k$.

\noindent
Towards (ii), let $(F,t)$ be
  consistent with $D$ and $\Sigma$.  The idea for the construction of \chseq is to inductively merge  all partial
linear chase sequences for nodes of $F$ in an inductive fashion. Let,
to this end, the nodes of $F$ be numbered $k_1,\ldots,k_r$. We let
$\chseq_1$ be the partial  linear chase sequence for $k_1$. We define
$\chseq_i$ by merging $\chseq_{i-1}$ with a  partial  linear chase
sequence $\chseq'$ for $k_i$ as follows. A complication is caused by
the fact that there might be facts $\At{f}{k_i}$, $\At{f'}{k_i}$ and
$\At{g}{k_j}$, for some $j<i$ such that $\At{f}{k_i}$ is needed to
produce $\At{g}{k_j}$ in $\chseq_{i-1}$  and $\At{g}{k_j}$  is needed to
produce $\At{f'}{k_i}$ in $\chseq'$. 
Therefore, we divide $\chseq'$ in
subsequences that end with a chase step that produces a fact from
$F$.\footnote{We can safely assume that the last step of $\chseq'$ is
  of this kind.} In the example, one subsequence would end producing
$\At{f}{k_i}$, one other producing $\At{f'}{k_i}$.
A subsequence producing a fact $\At{f}{k_i}$ in step $\ell$ of
$\chseq_i$ is then inserted
right after the maximal chase step of $\chseq_{i-1}$ that produces a
fact $g'$ of $F$ with $t(g')<s(j)$. 

The construction guarantees that the sequence \mbox{$\chseq_r$}
produces all facts from $F$.  Furthermore, it is guaranteed by the
timing functions that all witness facts that are used in chase steps
are produced in earlier steps. Altogether, \chseq is a chase sequence
that certifies $\Sigma\models\tau$. 

\noindent
Towards (iii), we sketch a nondeterministic polynomial space algorithm
that checks, given $\Sigma$, $D$ and $E$, whether there exists a  timed
witness set of the desired size that is consistent  with $D$ and
$\Sigma$ and contains $E$. This algorithm first guesses $(F,t)$,
such that $E\subseteq F$ and then checks that each node $k$ has
a partial linear chase sequence \chseq relative to $F-k$ with the
linear structure with respect to (G3)-predeccessors. For the latter,
the algorithm just guesses such a sequence step-by-step. Actually, if
such a sequence exists, there is always one of the following simple
form: it consists of  the composition of some subsequences, each of
which starts with a node-creating tgd and is continued by (zero or
more) collecting tgds for this node. Besides the first one, each
subsequence begins with a node-creating dtgd of type~(G3) which uses
the node of the previous series for the unbounded body variable. Thus,
whenever this sequence generates a new node $\ell'$ by applying a
(G3)-dtgd to a node $\ell$ not in $F$, it can forget $\ell$ and its
facts afterwards. Indeed, $\ell$ is not needed as a witness node
thanks to $F$ and does not generate any further nodes because of the
linear structure of (G3)-predecessors. %
This completes the proof of the upper bound.

  \medskip
  \noindent
  To prove
  Theorem~\ref{thm:ImplicationForWeaklyBoundedTransferIsPSPACEcomplete}.2,
  that is, \PSPACE-hardness of $\Implication(\Twbd)$ for fixed $\alpha
  \ge 1$ and $b \ge 0$, we sketch a reduction from the \PSPACE-hard
  word problem for linear bounded automata similar to the proof of
  \PSPACE-hardness of the implication problem for inclusion
  constraints \cite{DBLP:journals/jcss/CasanovaFP84}.
  Let $w$ be an input word $w$ of length $n$ (which for simplicity is assumed to
  carry border symbols left and right) over some alphabet $\Gamma$,
  which also contains all tape symbols of the automaton and let $Q$ be
  the state state set of the automaton with initial state $s$ and
  accepting state $h$. We use one unary relation symbol $P_{a,i}$ for each
  symbol $a\in\Gamma$ and each position $i$ in $w$. Furthermore, we
  use one unary relation symbol $S_q$, for each state $q$ of $Q$ and
  one  unary relation symbol $H_{i}$ for each position of $w$. The
  idea is to encode information about configurations by facts over
  $\{0,1\}$. That position $3$ of the tape carries symbol $b$ would be
  represented by fact $P_{b,3}(1)$ and facts $P_{a,3}(0)$, for all
  $a\not=b$. Each configuration occurring in the computation is
  represented by one node. The body of $\tau$ consists of the facts that represent
  the initial configuration of the automaton  on input $w$ at one node
  $k_0$. For
  instance, it contains facts $\At{S_s(1)}{k_0}$, $\At{S_q(0)}{k_0}$ for $q\not=s$, and
  $\At{H_1(1)}{k_0},\At{H_2(0)}{k_0},\ldots$, $\At{P_{w_1,1}(1)}{k_0}$
  and so on.

  For each transition $\delta$ of the automaton applicable to state
  $q$ and symbol $a$, each position $j$ of
  $w$ and each combination of two symbols $b,c$ for positions
  $j-1$ and $j+1$, there is a dtgd of the following kind:
  $\At{S_q(1)}{\kappa}, \At{H_j(1)}{\kappa},
  \At{P_{b,j-1}(1)}{\kappa}, \At{P_{a,j}(1)}{\kappa},
  \At{P_{c,j+1}(1)}{\kappa}, \At{Y}{\kappa} \to \At{Y}{\lambda},\At{Z}{\lambda}$.
  Here $Y$ is a set of atoms of the form $P_{a,i}(x_{a,i})$, for all
  $a$ and all $i\not\in\{j-1,j,j+1\}$ and of the form $H_{i}(0)$, for all
$i\not\in\{j-1,j,j+1\}$. The set $Z$ contains atoms that represent the
state of the automaton after applying $\delta$ and the symbols at
positions $j-1,j,j+1$. It is important that all variables in the head
of the rule appear in the body, i.e., it is indeed data-full. Clearly,
these dtgds are of Type (G3). 

Finally, the head of $\tau$ is just $\At{S_h(1)}{\lambda}$.

By induction on the number of computation steps, it is straightforward
to show that an instance that satisfies the body of $\tau$ and all
constraints from $\Sigma$ must contain nodes for each configuration of
the computation and thus the existence of a node with fact $S_h(1)$ is
implied if and only if the computation reaches a configuration with
state $h$.

      \end{proof}

	  \subsection{Proof details for classes with \texorpdfstring{\EXPTIME}{EXPTIME}-reasoning}\label{app:EXPTIMEReasoning}

\begin{proof}[Proof of Theorem~\ref{prop:lowerbounds}]
  \begin{new}
     That $\Implication(\Tdf)$ is  $\EXPTIME$-hard follows from the fact that the implication problem for full tgds is already \EXPTIME-complete \cite{DBLP:conf/stoc/ChandraLM81}. That is, \EXPTIME-completeness can already be 
  realized by dtgds without node variables. However, the proof of that result uses schemas of unbounded arity, and  the problem is easily seen to be in~\NP for schemas of bounded arity.\
 Thus, the challenge of the proofs here is to work with  fixed schemas of arity 2.

For all four fragments, the basic proof strategy is the same.
Let $L$ be an \EXPTIME-complete problem that is decided by some alternating Turing machine with linearly bounded space.
We prove the lower bound  by a polynomial reduction from $L$.

	We next give a description of the basic idea and the general framework.\footnote{Although the general approach of the reduction is similar to the one described by Calì et~al. in \cite{DBLP:journals/jair/CaliGK13} for weakly guarded tgds, it is significantly different, since we can not use relations that relate two (or more) node variables.}

	\ProofSubstep{Basic idea}
	Given a word~$w$, an instance $(\Sigma,\tau)$ for the implication problem is computed such that the chase process is intended to simulate the computation of $M$ on $w$. To this end, every node~$\nodek$ represents a configuration~$C(\nodek)$ of $M$. The constraints in $\Sigma$ are used to generate nodes that represent all possible configurations (with $|w|$ tape cells) and to \enquote{compute} which configurations are accepting. A configuration of $M$ is \emph{accepting} if
        \begin{itemize}
        \item it has an accepting state,
        \item it is universal and both its successor configurations are accepting, or
        \item it is existential and at least one of its successor configurations is accepting.
        \end{itemize}
        $M$ accepts $w$ if the initial configuration is accepting.
        
	\ProofSubstep{General Framework}
        We can assume that  $M$ is of the form  $(Q,A,(\delta_1,\allowbreak\delta_2),\allowbreak q_0,F)$, where $q_0\in Q$ is the initial state of $M$, $F \subseteq Q$ is the set of accepting states,  $A=\{a_1,\ldots,a_t\}$ is the alphabet, and the set $Q$ of states is partitioned into existential and universal states, $Q = \Qex \uplus \Qall$. Furthermore, we assume that, for each state~$q$ and each symbol~$a \in A$, there are two transitions $\delta_1(q,a)$ and~$\delta_2(q,a)$. For convenience, we assume that the initial state of $M$ is not accepting.

        Let~$w=w_1 \dots w_n$ be an input word for~$M$. We can assume that $M$ uses only $n+2$ cells of the tape, where the first and the last cell (with positions~$0$ and~$n+1$, respectively) are marked with special symbols $a_1=\lmark,a_t=\rmark\in A$ that are never altered by the transition functions.  
Configurations of $M$ on input $w$ can be represented by triples $(q,i,u)$ with $q\in Q$, $i\le n+1$ and $|u|=n$, where the tape content is $\lmark u\rmark$ and the head is at position $i$.

Nodes generated during the chase are supposed to encode configurations in the following way.  The data values in $D_\tau$ are intended to consist of elements $0,\ldots,n+1$ representing the positions of the tape and further elements $j_1,\ldots,j_t$, one for each symbol in~$A$.

The schema of $(\Sigma,\tau)$ uses two kinds of relation symbols, with the given intended meaning. The first kind of relation symbols is only used for global facts:

\begin{itemize}
\item $\relAlph(j)$: element $j$ represents a symbol from $A$.
\item $\relAlph_r(j)$: element $j$ represents $a_r$.
\item $\relSucc(i,j)$: position $j$ is the successor position of position $i$, that is, the position to the right of $i$. 
\end{itemize}
The relation $\relSucc$ shall define a successor relation on $0,\ldots,n+1$.

                The second kind of relation symbols is used for local facts with the intention to encode one configuration per node:
                \begin{itemize}
                \item $\relSym(i,j)$: position $i$  carries the symbol represented by $j$.
\item $\relState_q()$: the configuration has state $q$, for  $q\in Q$.
		\item $\relHead(i)$: the head of the $M$ is at position~$i$.
                \item $\relAcc()$: the represented configuration is accepting.
                \item $\relAcc_1()$ and $\relAcc_2()$ indicate that the first (and the second, respectively) successor configuration is accepting. 
                \end{itemize}
                
            More precisely, a configuration $C=(q,p,u)$ is supposed to be represented by a node $k$ with the following facts:
                \begin{itemize}
				\item $\relState_q()\at{k}, \relHead(p)\at{k}$,
				\item                   $\relSym(0,j_1)\at{k},\relSym(n+1,j_t)\at{k}$,
				\item                   $\relSym(1,\ell_1)\at{k},\ldots, \relSym(n,\ell_n)$, where each $\ell_i$ is the element $j_r$ with $u_i=a_r$.
\end{itemize}

The intention of the body of $\tau$ is to establish in $D_\tau$ a successor relation on $0,\ldots,n+1$ and some elements $j_1,\ldots,j_t$ that represent $a_1,\ldots,a_t$. Furthermore, it guarantees that the initial configuration of $M$ on input $w$ is represented on some node $\nodek$.

The details of the construction of $\Sigma$ and $\tau$ differ for the four considered constraint classes.

\medskip\noindent
Now, we are ready to prove statements~(a)~--~(d) of Theorem~\ref{prop:lowerbounds}.

  \ProofStep{Proof (a)}
  We show that $\Implication(\Tdf)$ is \EXPTIME-hard, even if restricted to node-creating dtgds of type~(G2) and data-collecting dtgds of type~(C3) without comparison atoms.
  We start with a description of the proof idea, which is followed by the details.

  \ProofSubstep{Proof idea}
Algorithm~\ref{alg:atm} describes a procedure that  is supposed to be mimicked by $(\Sigma,\tau)$. 
During the first phase,
\codelineref[lst:ATMSim:Phase1:start]{lst:ATMSim:AcceptingStates:end}, it
generates nodes that represent all\footnote{It is not tested whether a
  configuration can actually occur in the computation of $M$ on input
  $w$.} possible configuration triples $(q,p,u)$ and adds fact
$\relAcc()$ to all nodes representing a configuration with an
accepting state $q$. In the second phase,
\codelineref[lst:ATMSim:SuccessorConfig:start]{lst:ATMSim:Phase2:end}, the
additional information whether a configuration is accepting is
transmitted to  configurations~$C$ from successor
configurations~$C_j$,  ($C \vdash_j C_j$), with the help of collecting
dtgds.
\begin{algorithm}[tbp]
  \caption{ATM Simulation} \label{alg:atm}
  \begin{algorithmic}[1]
   \REQUIRE String $w$
   \STATE Add node $\nodek_0$ representing $(q_0,0,w)$ \label{lst:ATMSim:init}
   \FOR {each possible configuration $(q,i,v)$ where $|v|=|w|$} \label{lst:ATMSim:Phase1:start}
	   \STATE Add a node representing~$(q,i,v)$
   \ENDFOR \label{lst:ATMSim:Phase1:end}
   \FOR {each node $\nodek$} \label{lst:ATMSim:Phase2:start}\label{lst:ATMSim:AcceptingStates:start}
   \IF {$\nodek$ has an accepting state}
   \STATE Add $\relAcc()$ to $\nodek$
	\ENDIF
	\ENDFOR \label{lst:ATMSim:AcceptingStates:end}
  \REPEAT
  \FOR {each pair $(\nodek, \nodem)$ of nodes and $j\in\{1,2\}$} \label{lst:ATMSim:SuccessorConfig:start}
   \IF {$C(\nodek) \vdash_j C(\nodem)$ and $\relAcc()\at\nodem$}
   \STATE Add $\relAcc_j()$ to $\nodek$
   \ENDIF
   \ENDFOR \label{lst:ATMSim:SuccessorConfig:end}
   \FOR {each node $\nodek$ with state~$q$} \label{lst:ATMSim:AcceptanceCheck:start}
   \IF {$q$ is existential and $\relAcc_1()$ \emph{or} $\relAcc_2()$ holds on~$k$}
   \STATE Add $\relAcc()$ to $\nodek$
   \ENDIF
   \IF {$q$ is universal and $\relAcc_1()$ \emph{and} $\relAcc_2()$ hold on~$k$}
   \STATE Add $\relAcc()$ to $\nodek$
  \ENDIF
  \ENDFOR \label{lst:ATMSim:AcceptanceCheck:end}
   \UNTIL {no more changes} \label{lst:ATMSim:Phase2:end}
   \STATE Accept iff $\relAcc()\at\nodek_0$ \label{lst:ATMSim:test}
  \end{algorithmic}
\end{algorithm}

\ProofSubstep{Proof details}
       We first describe the construction of $\Sigma$ and $\tau$.
			
      The initial assignment for $\nodek_0$ (\codelineref{lst:ATMSim:init}) is done by \body{\tau}. The final test whether $\nodek_0$ is accepting (\codelineref{lst:ATMSim:test}) is done by \head{\tau}. The intermediate processing has to be taken care of by $\Sigma$.
 
				\ProofSubstep{Construction of~$\tau$}

                With the initial configuration $C_0=(q_0,0,w)$ of $M$ on input $w$, we associate the set~$\atoms_{C_0}$, which is the union of the following sets
				\begin{itemize}
					\item $\{\relState_{q_0}(), \relHead(x_0)\}$,
					\item $\{\relSym(x_0,y_1), \relSym(x_{n+1}, y_t)\}$, and
					\item $\{\relSym(x_i,y_{r}) \mid i \in \{1,\dots,n\}, w_i=a_{r}\}$.
                                        \end{itemize}

				The body of $\tau$ is  the union of the sets $\atoms_\relSucc$,  $\atoms_\relAlph$, and $\atoms_{C_0}\at\vark$ of atoms, where $\atoms_\relSucc$ establishes a linear order on the variables $x_0,\ldots,x_{n+1}$ (and will be used more often), $\atoms_\relAlph$ assigns the alphabet elements and~$C_0$ is the initial configuration. To this end, we let
\begin{itemize}
\item $\atoms_\relSucc= \{\relSucc(x_0,x_1),\ldots,\relSucc(x_n,x_{n+1})\}$, and
\item $\atoms_\relAlph = \{\relAlph_r(y_r), \relAlph(y_r) \mid r \in \{1,\dots,t\}\}$.
\end{itemize}
Recall that $a_1$ and $a_t$ are the special symbols $\lmark$ and $\rmark$, respectively.

The head of $\tau$ consists of the single atom $\relAcc()\at\kappa$. 

\ProofSubstep{Construction of~$\Sigma$}
The set $\Sigma$ %
is the disjoint union of sets $\Sigma_1$ and~$\Sigma_2$ reflecting the first and the second phase of the algorithm, respectively.

For the generation of nodes representing all possible configurations in the algorithm (\codelineref[lst:ATMSim:Phase1:start]{lst:ATMSim:Phase1:end}) subset~$\Sigma_1$ contains a node-creating dtgd~$\sigma_{q,p}$ of Type (G2) for every state $q \in Q$ and every position $p \in \{0,\dots,n+1\}$.%

Its body is $\atoms_\relSucc\cup\{\relAlph_1(y_1),\relAlph_t(y_t),\relAlph(z_1),\ldots,\relAlph(z_{n})\}$ and its head is $\atoms_{q,p}\at\vark\cup\{\relAcc()\at\vark\}$ if~$q$ is accepting and $\atoms_{q,p}\at\vark$ otherwise, where set $\atoms_{q,p}$ is defined as the union of the sets
\begin{itemize}
	\item $\{\relState_q(), \relHead(x_p)\}$,
	\item $\{\relSym(x_0,y_1), \relSym(x_{n+1},y_t)\}$ and
	\item $\{\relSym(x_i,z_i) \mid i \in \{1,\dots,n\}\}$.
\end{itemize}
In particular, constraint $\sigma_{q,p}$ also takes care of \codelineref[lst:ATMSim:AcceptingStates:start]{lst:ATMSim:AcceptingStates:end}.

Subset~$\Sigma_2$ consists of all other constraints, defined in the following.

For \codelineref[lst:ATMSim:SuccessorConfig:start]{lst:ATMSim:SuccessorConfig:end}, there is one data-collecting dtgd of Type (C3), for each $q\in Q$,  $a_r\in A$, $j\in\{1,2\}$ and $p\in\{0,\ldots,n+1\}$ representing the $j$-th possible transition of $M$ in case its current state is $q$, the current tape symbol is $a_r$ and the current head position is $p$.

Let us assume that $\delta_j(q,a_r)$ requires that the current symbol is replaced by $a_s$, the head moves to the right, and the new state is $q'$. Then a dtgd exists if $p\le n$. Its body is the union of the following sets

\begin{itemize}
	\item $\atoms_\relSucc \cup \{\relAlph_r(z_p), \relAlph_s(z'_p)\}$,
	\item $\atoms_{q,p}\at\vark$,
	\item $\atoms_{q',p+1}\at\varm \setminus \{\relSym(x_p,z_p)\at\varm\}$ and
	\item $\{\relSym(x_p,z'_p), \relAcc()\}\at\varm$.
\end{itemize}

Thus, the intention of the body is to express that some node $\nodem$ encodes the $j$-th successor configuration~$C(m)$ of the configuration~$C(\nodek)$ of $\nodek$ and $C(m)$ is accepting.
The head of the dtgd thus just consists of $\relAcc_j()\at\vark$. 

The dtgds for other transitions are defined analogously.

Finally, to simulate \codelineref[lst:ATMSim:AcceptanceCheck:start]{lst:ATMSim:AcceptanceCheck:end} of the algorithm, for each existential state $q\in\Qex$ there are the dtgds
\begin{itemize}
\item $\relState_q()\at\kappa, \relAcc_1()\at\kappa \to \relAcc()\at\kappa$ and
\item $\relState_q()\at\kappa, \relAcc_2()\at\kappa \to \relAcc()\at\kappa$,
\end{itemize}
and for each universal state $q\in\Qall$ there is the dtgd
$$
\relState_q()\at\kappa, \relAcc_1()\at\kappa, \relAcc_2()\at\kappa \to \relAcc()\at\kappa.
$$

	\ProofSubstep{Correctness}
	We claim that $\Sigma\models\tau$ if and only if $M$ accepts the input word~$w$ encoded by~$\tau$. More precisely, starting from a canonical database $D_\tau$ the chase  procedure generates exactly the same nodes as Algorithm~\ref{alg:atm}, modulo renaming of elements.

        First of all, $\body{\tau}$ ensures that the canonical database $D_\tau$ consists of one node $k_0$ representing $C_0$ and thus takes care of \codelineref{lst:ATMSim:init}. The correspondence of the other parts of Algorithm~\ref{alg:atm} to the constraints of $\Sigma$ was already described above.

        It is straightforward to show by induction that, for every node produced by Algorithm~\ref{alg:atm}, a corresponding node with the same facts is generated by the chase, and vice versa. Finally $\head{\tau}$ is implied by $\Chase(\chseq,D_\tau)$ if and only if $M$ accepts $w$.

       \ProofStep{Proof of (b)}
	   We show that $\Implication(\Tdf)$ is \EXPTIME-hard, even if restricted to node-creating dtgds of type~(G2) and type~(G4) without comparison atoms.
  We start with a description of the proof idea, which is followed by the details.

	  \ProofSubstep{Proof idea}
        The proof mainly differs from the previous one in the way in which
        the information about accepting  configurations is propagated. Each configuration is represented by up to four nodes~$\nodek,\nodek_1,\nodek_2$ and~$\nodek^*$ that differ only with respect to acceptance facts. If, for nodes $\nodek$ and $\nodem$, it holds $C(\nodek) \vdash_j C(\nodem)$ and \nodem contains fact $\relAcc()$, then a new node $\nodek_j$ is generated with all facts of $\nodek$ plus the additional fact  $\relAcc_j()$. If a configuration $C$ with a universal state is represented by nodes $\nodek_1,\nodek_2$, and $\nodek_1,\nodek_2$ contain facts $\relAcc_1()$ and $\relAcc_2()$, respectively, then another node $\nodek^*$ is generated which also represents $C$ and has the additional fact $\relAcc()$. Similarly for configurations with existential states.

       \ProofSubstep{Proof details}
 	Let $L$, $M$ and $w$ be as in (a).

        The goal of the construction is to guarantee that a node that represents the initial configuration $C_0$ and has fact $\relAcc()$ is generated, if and only if $M$ accepts $w$.
        
	\ProofSubstep{Construction}
        $\Sigma=\Sigma_1 \uplus \Sigma_2$ and $\tau$ are again constructed to make the chase simulate an algorithm very similar to Algorithm~\ref{alg:atm}. 

	\ProofSubstep{Construction of~$\tau$}
	The body of~$\tau$ is the same as in~(a) and its head is $\atoms_{C_0}\at\vark^* \cup \{\relAcc()\at\vark^*\}$.

	\ProofSubstep{Construction of~$\Sigma$}
	Subset~$\Sigma_1$ is defined exactly as in the proof of (a).

	Subset~$\Sigma_2$ consists of  the following constraints.

To generate nodes~$\nodek_j$ along the lines sketched above, $\Sigma_2$ has one node-creating dtgd of Type (G4), for every $j \in \{1,2\}$, every $q \in Q$, every $p \in \{1,\dots,n+1\}$ and every symbol~$a_r \in A$, depending on the transition~$\delta_j(q,a_r)$. As an example we give the dtgd for a transition that replaces~$a_r$ by~$a_s$, moves the head to the right (assuming $i \le n$) and enters state~$q'$. The body of this dtgd is the union of the sets
	\begin{itemize}
		\item $\atoms_\relSucc \cup \{\relAlph_r(z_p), \relAlph_s(z'_p)\}$,
		\item $\atoms_{q,p}\at\vark$,
		\item $\atoms_{q',p+1}\at\varm \setminus \{\relSym(x_p,z_p)\at\varm\}$ and
		\item $\{\relSym(x_p,z'_p),\relAcc()\}\at\varm$.
	\end{itemize}
	The head of the dtgd is $\atoms_{q,p}\at\vark_j \cup \{\relAcc_j()\at\vark_j\}$.

	Thus, a new node~$\nodek_j$ is generated with an $\relAcc_j()$ fact in the local instance if there is a node~\nodek representing the same configuration $C(\nodek)=C(\nodek_j)$ and successor configuration~$C(m)$ for $C(\nodek) \vdash_j C(\nodem)$ is marked accepting on some node~\nodem. Dtgds for other transitions are defined analogously.

	Furthermore, for each universal state~$q$ and each $p \in \{0,\dots,n+1\}$, there is a node-creating dtgd of Type (G4) whose body is the union of
	\begin{itemize}
		\item $\atoms_{q,p}\at\vark_1 \cup \{\relAcc_1()\at\vark_1\}$ and
		\item $\atoms_{q,p}\at\vark_2 \cup \{\relAcc_2()\at\vark_2\}$
	\end{itemize}
	and whose head is $\atoms_{q,p}\at\vark^* \cup \{\relAcc()\at\vark^*\}$.
	Similarly, for each existential state~$q$ and each $p \in \{0,\dots,n+1\}$, there are two node-creating dtgds, one for each $j \in \{1,2\}$.  The body of the $j$-th dtgd is $\atoms_{q,p}\at\vark_j \cup \{\relAcc_j()\at\vark_j\}$ and the head is $\atoms_{q,p}\at\vark^* \cup \{\relAcc()\at\vark^*\}$.

	\ProofSubstep{Correctness}
It is not hard to show  for each configuration $C$, that $C$ is accepting if and only if there is a node $k$ that represents $C$ and contains fact $\relAcc()$.

\ProofStep{Proof of (c)}
	We show that $\Implication(\Tdf)$ is \EXPTIME-hard, even if restricted to node-creating dtgds of type~(G2) and degds of type~(E4) without comparison atoms.
  We start with a description of the proof idea, which is followed by the details.

	\ProofSubstep{Proof idea}
        The proof follows a similar strategy as the proof of (b).   However, the goal of the construction is slightly different. Instead of propagating acceptance information by the generation of new nodes, this construction propagates information with the help of degds. A degd can identify two nodes and thus yield a node that contains the facts of \emph{both nodes}.

        To this end, the chase first generates three nodes $\nodek,\nodek_1,\nodek_2$, for each possible configuration with the initial additional facts $\relEval()$, $\relAcc_1()$, and $\relAcc_2()$, respectively.
        The node $\nodek$ with fact  $\relEval()$ is supposed to collect the acceptance information for its configuration. More precisely, if $C(k)\vdash_1 C(m)$, then $\nodek$ and $\nodek_1$ are identified yielding a node with $\relEval()$ \emph{and} $\relAcc_1()$. Likewise for $C(k)\vdash_2 C(m)$, $\nodek$ and $\nodek_2$. The fact $\relAcc()$ is added according to the semantics of universal or existential nodes.
        
        Altogether, the chase should generate a node that represents the initial configuration $C_0$ and has facts $\relEval()$ \emph{and} $\relAcc()$,  if and only if $M$ accepts $w$.

	\ProofSubstep{Proof details}
	Let $L$, $M$ and $w$ be as in (a).

	\ProofSubstep{Construction of~$\tau$}
        The body of $\tau$ is defined as in (a). Its head is the union of~$\atoms_{C_0}\at\vark$ and $\{\relEval(), \relAcc()\}\at\vark$.

	\ProofSubstep{Construction of~$\Sigma$}
	Subset~$\Sigma_1$ contains three node-creating dtgds for every $q \in Q$ and every $p \in \{0,\dots,n+1\}$. They all share the same body, which is the union of sets $\atoms_\relSucc$ and $\atoms_\relAlph$. The head of each dtgd is the union of $\atoms_{q,p}\at\vark$ with $\{\relEval()\at\vark\}$, $\{\relAcc_1()\at\vark\}$, and $\{\relAcc_2()\at\vark\}$, respectively.

	Furthermore, for every $q \in Q$, every $j \in \{1,2\}$, every $p \in \{0,\dots,n+1\}$ and every $a_r \in A$, a degd of Type (E4) that depends on the transition $\delta_j(q,a_r)$ is added. We exemplify this for a transition that replaces the current symbol by~$a_s$, moves right and enters state~$q'$. The head of the degd is $\vark=\vark_j$. The body is the union of sets
	\begin{itemize}
		\item $\calA_\relSucc$,
		\item $\atoms_{q,p}\at\vark \cup \{\relEval()\at\vark\}$,
		\item $\atoms_{q,p}\at\vark_j \cup \{\relAcc_j()\at\vark_j\}$,
		\item $\big(\atoms_{q',p+1} \setminus \{\relSym(x_p,z_p)\}\big)\at\varm \cup \{\relAcc()\at\varm\}$ and
		\item $\{\relSym(x_p,z'_p), \relAlph_s(z'_p)\}\at\varm$.
	\end{itemize}
	Degds for other transitions are defined analogously.

	Additionally, for each universal state~$q$, there is a data-collecting dtgd
        \begin{itemize}
       \item $\relState_q()\at\kappa, \relAcc_1()\at\kappa, \relAcc_2()\at\kappa \to \relAcc()\at\kappa$ 
        \end{itemize}
 and for each existential state~$q$, there are two data-collecting dtgds
       \begin{itemize}
       \item $\relState_q()\at\kappa, \relAcc_1()\at\kappa, \to \relAcc()\at\kappa$ 
       \item $\relState_q()\at\kappa, \relAcc_2()\at\kappa \to \relAcc()\at\kappa$.
        \end{itemize}

        	\ProofSubstep{Correctness}
                It is straightforward to show by induction that a node representing a configuration $C$ and containing facts $\relEval()$ and  $\relAcc()$ is generated if and only if $C$ is accepting.

         \ProofStep{Proof of (d)}
		 We show that $\Implication(\Tdf)$ is \EXPTIME-hard, even if restricted to node-creating dtgds of types~(G2) and~(G3) and degds of type~(E3) without comparison atoms.
  We start with a description of the proof idea, which is followed by the details.

	  \ProofSubstep{Proof idea}
        Again, the proof strategy is similar to the previous proofs. However, it differs in that it does not start by generating nodes for all possible configurations, but only for those with accepting states. These nodes have, in particular, the  fact $\relAcc()$.

        If $C \vdash_1 C(m)$ and $m$ carries $\relAcc()$ then a new node $\nodek$ with $C(\nodek)=C$ is generated which contains also $\relAcc_1()$. Likewise for $C \vdash_2 C(m)$. Then, if there are two nodes $\nodek_1$ and $\nodek_2$ with facts  $\relAcc_1()$ and  $\relAcc_2()$, respectively, which both represent the same configuration, they are identified yielding one node representing $C$ and containing $\relAcc_1()$ \emph{and}  $\relAcc_2()$. Then $\relAcc()$ can be added, as before.
        
        Altogether, a node that represents the initial configuration $C_0$ and has facts $\relAcc()$ should be generated if and only if $M$ accepts $w$. 

\end{new}
	\ProofSubstep{Proof details}
	Let $L$, $M$ and $w$ be as in (a).

	\ProofSubstep{Construction of~$\tau$}
	The constraint $\tau$ is  defined just like in (a).

	\ProofSubstep{Construction of~$\Sigma$}
        Subset $\Sigma_1$ contains a node-creating dtgd~$\sigma_{q,p}$ of Type (G2) for every \emph{accepting} state $q \in F$ and every position $p \in \{0,\dots,n+1\}$.

        Subset $\Sigma_2$ has a dtgd of Type (G3) for each $q \in Q$, each $j \in \{1,2\}$, each $p \in \{0,\dots,n+1\}$ and each $a_r \in A$. We illustrate the definition of this dtgd for a transition $\delta_j(q,a_r)$ that replaces the current symbol by~$a_s$, moves to the right and enters state $q'$. In this case, the body of the dtgd is $\atoms_{q',p+1}\at\varm \cup \{\relAlph_s(z_p)\at\varm, \relAcc()\at\varm\}$ and its head is the union of
	\begin{itemize}
		\item $\atoms_{q,p}\at\vark \setminus \{\relSym(x_p,z_p)\at\vark\}$ and
		\item $\{\relSym(x_p,z'_p)\at\vark, \relAlph_r(z'_p)\at\vark, \relAcc_j()\at\vark\}$.
	\end{itemize}
	Dtgds for other transitions are defined in an analogous fashion.

        To identify the two nodes representing the same configuration, $\Sigma_2$ contains a degd of Type (E3), for every $q \in Q$ and every $p \in \{1,\dots,n+1\}$. Its head is $\vark=\varm$ and its body is $\atoms_{q,p}\at\vark \cup \atoms_{q,p}\at\varm$. We note that, although the facts $\relAcc_1()$ and  $\relAcc_2()$ do not occur in this degd, the only way in which two nodes can be identified is, if they represent the same configuration, one contains $\relAcc_1()$ and  the other contains $\relAcc_2()$.

	Finally, $\Sigma_2$ contains the same data-collecting dtgds as in (c) to infer $\relAcc()$-facts.

	\ProofSubstep{Correctness}
                It is again straightforward to show by induction that a node representing a configuration $C$ and containing facts $\relEval()$ and  $\relAcc()$ is generated, if and only if $C$ is accepting.
\end{proof}

\end{document}